\documentclass[11pt]{article}
\usepackage{fullpage}

\usepackage{dsfont}
\newcommand{\citeyearp}[1]{[\citeyear{#1}]}

\usepackage{url}
\usepackage{graphicx}
\usepackage{epstopdf}
\usepackage{enumitem}
\usepackage{wrapfig}
\setlist{itemsep=0pt, topsep=2pt}

\usepackage[table,xcdraw]{xcolor}

\usepackage{algorithmic}
\usepackage{algorithm}
\usepackage{caption}
\usepackage{float}
\ifpdf
  \DeclareGraphicsExtensions{.eps,.pdf,.png,.jpg}
\else
  \DeclareGraphicsExtensions{.eps}
\fi

\usepackage[utf8]{inputenc}
\usepackage{mathtools}
\usepackage{amsmath}
\usepackage{amsfonts,amsthm,boxedminipage,color,url}
\usepackage{fullpage}
\usepackage{amssymb}
\usepackage{natbib}
\usepackage{soul}
\usepackage{enumitem}
\usepackage{booktabs}
\usepackage{multirow}
\usepackage[pdfstartview=FitH,colorlinks,linkcolor=blue,filecolor=blue,citecolor=blue,urlcolor=black]{hyperref}
\usepackage{hyperref} 
\usepackage{aliascnt,cleveref}

\newcommand{\email}[1]{\href{mailto:#1}{#1}}

\usepackage{amsopn}

\DeclareMathOperator*{\argmax}{arg\,max}

\usepackage{tikz}
\usetikzlibrary{decorations.pathreplacing}
\usetikzlibrary{patterns}
\usepackage{pgfplots}
\pgfplotsset{compat=1.18}
\usepackage{caption}
\usepackage{subcaption}

\usepackage{xcolor}
\pgfplotsset{compat=1.18}
\pgfplotsset{compat=newest}
\usetikzlibrary{positioning}

\newcommand{\agents}{\mathcal{N}}
\newcommand{\actions}{\mathcal{A}}
\newcommand{\set}{S}

\newcommand{\cost}{c}
\newcommand{\reals}{\mathbb{R}}
\newcommand{\contract}{\boldsymbol{\alpha}}
\newcommand{\Seq}{S^\dagger}
\newcommand{\opt}{OPT}
\newcommand{\sstar}{S^{\star}}
\newcommand{\downscaling}{\Psi}

\newcommand{\ind}[1]{\mathds{1}[#1]}

\makeatletter
\def\moverlay{\mathpalette\mov@rlay}
\def\mov@rlay#1#2{\leavevmode\vtop{%
   \baselineskip\z@skip \lineskiplimit-\maxdimen
   \ialign{\hfil$\m@th#1##$\hfil\cr#2\crcr}}}
\newcommand{\charfusion}[3][\mathord]{
    #1{\ifx#1\mathop\vphantom{#2}\fi
        \mathpalette\mov@rlay{#2\cr#3}
      }
    \ifx#1\mathop\expandafter\displaylimits\fi}
\makeatother

\newcommand{\bigcupdot}{\charfusion[\mathop]{\bigcup}{\cdot}}

\definecolor{SetA}{RGB}{0,0,255}     % Blue
\definecolor{SetB}{RGB}{0,160,0}     % Green
\definecolor{SetC}{RGB}{255,0,0}     % Red
\definecolor{SetD}{RGB}{160,0,160}   % Violet
\definecolor{SetE}{RGB}{100,100,100} % Gray
\definecolor{SetF}{RGB}{0,128,128} % Teal
\newcommand{\lb}{L}
\newcommand{\rb}{R}
\newcommand{\intersection}{M}

\newtheorem{theorem}{Theorem}[section]
\newtheorem{lemma}[theorem]{Lemma}

\newtheorem{proposition}[theorem]{Proposition}
\newtheorem{definition}[theorem]{Definition}
\newtheorem{example}[theorem]{Example}
\newtheorem{remark}[theorem]{Remark}

\usepackage[]{color-edits}% suppress
\addauthor{MF}{red}

\addauthor{PD}{orange}

\addauthor{OK}{blue}

\addauthor{MS}{olive}

\addauthor{YG}{purple}

\addauthor{TP}{magenta}

\begin{document}

% Replace URL with link to full paper or comment out this line

\title{Combinatorial Contract Design:\\
Recent Progress and Emerging Frontiers\thanks{
This project has been partially funded by the European Research Council (ERC) under the European Union's Horizon Europe Program (grant agreement No.~101170373), by an Amazon Research Award, by the NSF-BSF (grant no.~2020788), and by the Israel Science Foundation Breakthrough Program (grant No.~2600/24), and by a grant from TAU Center for AI and Data Science. 

The author is profoundly grateful to all co-authors whose collaborations are surveyed in this article, and in particular to Paul D\"utting, Tomer Ezra, Yoav Gal-Tzur, Thomas Kesselheim, Tomasz Ponitka, Maya Schlesinger and Inbal Talgam-Cohen for generous feedback and thoughtful suggestions during its preparation.
}}

\author{Michal Feldman\thanks{Blavatnik School of Computer Science and Artificial Intelligence, Tel Aviv University (\email{mfeldman@tau.ac.il}, \url{https://www.mfeldman.sites.tau.ac.il/}).}}

\date{}

\maketitle

\begin{abstract} 
Contract theory studies how a principal can incentivize agents to exert costly, unobservable effort through performance-based payments. While classical
economic models provide elegant characterizations of optimal solutions, modern applications, ranging from online labor markets and healthcare to AI delegation and blockchain protocols, call for an algorithmic perspective. The challenge is no longer only which contracts induce desired behavior, but whether such contracts can be computed efficiently. This viewpoint has given rise to \emph{algorithmic contract design}, paralleling the rise of algorithmic mechanism design two decades ago.

This article focuses on \emph{combinatorial contracts}, an emerging frontier within algorithmic contract design, where agents may choose among exponentially many combinations of actions, or where multiple agents must work together as a team, and the challenge lies in selecting the right composition. These models capture a wide variety of real-world contracting environments, from hospitals coordinating physicians across treatment protocols to firms hiring teams of engineers for interdependent tasks. We review three combinatorial settings: (i) a single agent choosing multiple actions, (ii) multiple agents with binary actions, and (iii) multiple agents each selecting multiple actions. For each, we highlight structural insights, algorithmic techniques, and complexity barriers. Results include tractable cases such as gross substitutes reward functions, hardness results, and approximation guarantees under value- and demand-oracle access.
By charting these advances, the article maps the emerging landscape of combinatorial contract design, and highlights fundamental open questions and promising directions for future work.
\end{abstract}

\section{Introduction}
\label{sec:intro}
At the heart of contract theory lies a fundamental tension: a principal delegates a task to an agent, yet the agent's effort and behavior are not directly observable. Performance-based payment schemes --- or contracts --- serve as the instrument for aligning incentives. Classic economic theory offers a powerful framework for this principal–agent problem, offering deep insights into moral hazard and incentive provision~\citep{GrossmanHart83,Holmstrom79,milgrom-holmstrom87,HolmstromMilgrom91,Salanie17,Mirrlees75}, a body of work recognized by the 2016 Nobel Prize awarded to 
Hart and Holmstr\"om
~\citeyearp{Nobel}.

In recent years, this subject has taken on new life at the interface of economics and computation. As interactions increasingly occur in computational environments --- online labor platforms \citep{KaynarS22,HoSV16}, influencer marketing \citep{DecarolisGP20}, healthcare services \citep{BastaniGB18}, delegating machine learning tasks~\citep{CaiDP15} or blockchain protocols \citep{EyalS18,CongHe19} --- the question is no longer only which contracts induce desired behavior, but also whether they can be computed efficiently. This computational perspective has given rise to \emph{algorithmic contract design}, a growing research frontier that parallels the rise of algorithmic mechanism design two decades ago. 
Moreover, as AI agents take on increasingly complex tasks, questions of incentive alignment and computational tractability are only becoming more pressing \citep{Hadfield-Menell19a,WangEtAl2023,SaigET24}.

This research agenda was pioneered by~\citet{BabaioffFN06,BabaioffFNW12}, 
and further shaped by early influential contributions such as~\citep{HoSV16} and~\citep{DuttingRT19}.
The analogy with mechanism design is instructive. Mechanism design asks how to allocate resources and set payments when agents have private information. Its algorithmic counterpart reshaped the field by mapping tractability frontiers and providing approximation guarantees where exact solutions were infeasible. Contract theory addresses an orthogonal challenge: not eliciting private information but incentivizing costly, hidden actions. Bringing algorithmic tools to contract theory similarly opens new directions --- highlighting structural insights and charting the computational landscape of finding optimal (or approximately optimal) contracts.

\paragraph{From simple to combinatorial contracts.}
The simplest principal-agent model involves a single agent, choosing a single action from a set of available actions, with the optimal contract derived by linear programming. But practical applications rarely stay within such clean boundaries. 
In reality, employers, platforms, and policymakers regularly face situations where incentives must be coordinated across a team, across tasks, or across both. Multiple agents interact, decisions involve collections of actions, and outcomes may depend on complex combinations of efforts, taken by different agents.

Consider a hospital contracting with a team of physicians, where each physician chooses among multiple treatment protocols; or a tech company hiring several engineers, each able to pursue different combinations of development tasks; or a logistics firm employing couriers, each choosing routes composed of multiple delivery stops. In all these examples, the contract must account not only for who exerts effort, but also which combination of actions they undertake. Such settings naturally lead to what we call combinatorial contracts.

The combinatorial viewpoint dramatically enlarges the design space. Even with a single agent facing multiple possible actions, the principal must consider tradeoffs across a potentially exponential set of action combinations. With multiple agents, the challenge expands further: an agent's best response now depends on the actions of others, raising coordination issues. This interdependence creates additional obstacles, such as the risk of agents ``free riding'' on the effort of their teammates, making it more difficult to align incentives. 

The richness of combinatorial settings brings both challenges and opportunities. The primary challenge is computational: designing contracts over exponentially large action and agent spaces calls for new algorithmic tools and techniques. At the same time, viewing these problems through a computational lens uncovers structural properties that would otherwise remain hidden. These insights, in turn, form the basis for algorithmic techniques and approximation schemes that achieve strong performance guarantees. In this way, combinatorial contracts are emerging as a central paradigm in algorithmic contract design, much as combinatorial auctions have in algorithmic mechanism design~\citep{nisan2007combinatorial}.

\paragraph{This article.}
This article focuses exclusively on combinatorial contracts, emphasizing two prominent dimensions that have received substantial attention in recent years: multiple actions and multiple agents. 
We begin in Section~\ref{sec:model} with combinatorial preliminaries that are relevant across all models. We then review three settings, each treated in a dedicated section: Section~\ref{sec:model1} examines the single-agent, multiple-action model, where an agent chooses subsets of actions; Section~\ref{sec:model2} turns to the multi-agent, binary-action model, where each agent decides whether to exert costly effort and outcomes depend on the collective behavior of the group; and Section~\ref{sec:model3} combines the two dimensions in the multi-agent, multi-action model.

Our goal is to summarize what is currently known about these models: which reward structures admit tractable algorithms, what hardness barriers arise, and what approximation guarantees can be established. By mapping these results, we aim to chart the emerging landscape of combinatorial contracts and provide a foundation for further exploration of this new frontier in the intersection of economics and computation. 

Notably, there are additional directions within combinatorial contract design that are equally important, but which we do not cover here. These include models with multiple outcomes~\citep{DuttingRT21} or multiple principals~\citep{AlonLST23}, both of which introduce significant new incentive and computational challenges.

Finally, an earlier survey of algorithmic contract design appeared in~\citep{DuttingFT24-survey}, offering a broad overview of the field. The present article complements that work by providing a more detailed treatment of combinatorial contracts, with particular attention to the multi-action and multi-agent dimensions.

\section{Combinatorial optimization preliminaries}
\label{sec:model}

The study of combinatorial contracts naturally builds on tools from combinatorial optimization. Before turning to the specific models of contract design, we review the key definitions and algorithmic primitives that will be used throughout the survey.

\subsection{Set functions}
\label{sec:set-functions}

Let $U=[n]$ be a ground set of $n$ elements.  
A \emph{set function} is a mapping $f:2^U \to \reals$ that assigns a real value to each subset of $U$, with $f(S)$ denoting the value associated with $S \subseteq U$.  
In what follows, the notation $[n]$ will be used for different underlying sets, such as the action space in Section~\ref{sec:model1} or the set of agents in Section~\ref{sec:model2}.

We will restrict attention to set functions that are \emph{normalized}, meaning $f(\emptyset)=0$, and \emph{monotone}, i.e., whenever $S \subseteq T \subseteq U$, it holds that $f(S) \leq f(T)$.

For two sets $S,T \subseteq U$, the \emph{marginal value} of $S$ relative to $T$ is defined as $f(S \mid T) = f(S \cup T) - f(T)$.
If $S$ is a singleton $\{j\}$, we write $f(j \mid T)$ for convenience.  

In this survey, we focus primarily on well-studied families of set functions, especially those in the hierarchy of complement-free valuations introduced by \cite{LehmannLN06} in the study of combinatorial auctions (see also~\citep{BlumrosenN06}), while occasionally considering classes that allow for complementarities.

\begin{definition}
Let $U$ be a set of size $n$. 
A set function 
$f: 2^U \rightarrow \reals_{\geq 0}$
is said to be: 
\begin{itemize}
    \item \emph{Additive} if there exist $f_1,\ldots,f_n$
    such that $f(S)=\sum_{i\in S}f_i$ for every set $S \subseteq U$.
    \item \emph{Gross substitutes (GS)} if it is submodular (see below) and it satisfies the following triplet condition: for any set $S \subseteq U$, and any three elements $i,j,k \not\in S$, it holds that 
    \[
    f(i \mid S) + f(\{j,k\} \mid S) \leq \max\left(f(j \mid S) + f(\{i,k\} \mid S), f(k \mid S) + f(\{i,j\} \mid S)\right).
    \vspace{-0.4cm}
    \]
    \item \emph{Submodular} if for any two sets $S \subseteq T \subseteq U$, and any element $j \not\in T$, $f(j \mid T) \leq f(j \mid S)$. This class captures the property of diminishing marginal contribution.
    \item \emph{XOS} if it is a maximum over additive functions. 
    That is, there exists a set of additive functions $f_1,\ldots, f_{\ell}$ such that for every set $S \subseteq U$, 
    $
    f(S)=\max_{i \in [\ell]} \left( f_i(S)\right)
    $.
    \item \emph{Subadditive} if for any two sets $S, T \subseteq U$, it holds that $f(S)+f(T) \geq f(S \cup T)$.
    \item \emph{Supermodular} if for any two sets $S \subseteq T \subseteq U$, and any action $j \not\in T$, $f(j \mid T) \geq f(j \mid S)$.
\end{itemize}
\label{def:classes}
\end{definition}

With the exception of the supermodular class, all of the above families are \emph{complement-free}.  
It is a classical result \citep{LehmannLN06} that these families form a strict hierarchy:
\[
\text{Additive} \;\subset\; \text{GS} \;\subset\; \text{Submodular} \;\subset\; \text{XOS} \;\subset\; \text{Subadditive}.
\]

\subsection{Oracle access to set functions}
\label{sec:oracle}

Since explicitly representing $f$ requires exponential size, it is standard to access the function through oracle queries.  
Two types of queries are commonly used:  
\begin{itemize}
    \item \emph{Value query:} input is a set $S \subseteq U$; output is its value $f(S)$.  
    \item \emph{Demand query:} input is a nonnegative price vector $p = (p_1,\ldots,p_n) \in \mathbb{R}^n_{\ge 0}$; output is a set $S \subseteq U$ that maximizes $f(S) - \sum_{i \in S} p_i$.
\end{itemize}
The collection of all sets maximizing $f(S)-\sum_{i \in S}p_i$ is called the \emph{demand} of $f$ under prices $p$, and is denoted by  $D_f(p)=\argmax_{S \subseteq \actions} \{f(S)-\sum_{i \in S}p_i\}$.
In the context of combinatorial markets, the demand is the collection of the most preferred bundles of goods, given a valuation function $f$ (over bundles of goods) and item prices $p_1,\ldots,p_n$.

Demand queries are known to be stronger than value queries. Indeed, one can simulate value queries using poly-many demand queries, but not the other way around: for many natural classes of set functions, one cannot answer a demand query with poly-many value queries, under standard complexity assumptions~\citep{BlumrosenN09}.

\subsection{Demand oracle for gross substitutes (GS) functions}
\label{sec:demand-gs}

One exception to the hardness of demand queries is the gross substitutes (GS) class, where it is well-known that a demand query can be solved using the greedy procedure described in Algorithm~\ref{alg:gs-demand-query}. This is in fact a characterization of GS functions. 
(\Cref{alg:ultra-demand-query} is a variant of~\Cref{alg:gs-demand-query}, which we will return to in~\Cref{sec:ultra}.)

\begin{theorem}[GS characterization~\citep{PaesLeme17}]
    \label{thm:gs-greedy}
    A function $f$ is GS if and only if, for any price vector $p$, a demand query for $(f,p)$ can be solved using Algorithm \textsf{GreedyGS} (Algorithm~\ref{alg:gs-demand-query}). 
\end{theorem}

\begin{center}
\begin{minipage}[t]{0.48\linewidth}
\begin{algorithm}[H]
\caption{\textsf{GreedyGS}$(f,p)$}
\label{alg:gs-demand-query}
\begin{algorithmic}[1]
\STATE $S \gets \emptyset$
\FOR{$i = 1$ \TO $n$}
  \STATE \label{st:gs-demand-select}
    choose $x_i \in \argmax_{x \notin S} \{\, f(x \mid S) - p(x) \,\}$
  \IF{$f(x_i \mid S) - p(x_i) \le 0$}
    \STATE \textbf{return} $S$
  \ENDIF
  \STATE $S \gets S \cup \{x_i\}$
\ENDFOR
\STATE \textbf{return} $S$
\end{algorithmic}
\end{algorithm}
\end{minipage}
\hfill
\begin{minipage}[t]{0.48\linewidth}
\begin{algorithm}[H]
\caption{\textsf{GreedyUltra}$(f,p)$}
\label{alg:ultra-demand-query}
\begin{algorithmic}[1]
\STATE $S^* \gets \emptyset$;\; $S_0,S_1,\ldots,S_n \gets \emptyset$
\FOR{$i = 1$ \TO $n$}
    \STATE \label{st:alg-ultra-demand-query-i}
    choose $x_i \in \argmax_{x \notin S_{i-1}} \{\, f(x \mid S_{i-1}) - p(x) \,\}$
    \STATE $S_i \gets S_{i-1} \cup \{x_i\}$
\ENDFOR
\STATE \label{st:alg-ultra-demand-query-final}
choose $S^* \in \argmax_{i \in [n]} \{\, f(S_i) - p(S_i) \,\}$
\STATE \textbf{return} $S^*$
\end{algorithmic}
\end{algorithm}
\end{minipage}
\end{center}

\subsection{Approximation algorithms}
\label{sec:approximation-defs}

Throughout this survey we will be interested in the algorithmic tractability of contract design problems. 
Since many of these problems are computationally hard, it is natural to ask for approximation algorithms. 
We briefly recall the standard definitions here, phrased in terms of maximizing some objective function (such as the principal's expected utility, formally introduced later).

A (randomized) algorithm is said to provide a \emph{$\rho$-approximation} (where our convention 
is
that $\rho \ge 1$) if the (expected) value of the solution it finds is at least $\frac{1}{\rho}\opt$, where $\opt$ denotes the optimal value of the objective.

A \emph{fully polynomial-time approximation scheme (FPTAS)} is an algorithm that provides a multiplicative $(1+\epsilon)$-approximation, in time polynomial in the input size and $1/\epsilon$.  
A \emph{polynomial-time approximation scheme (PTAS)} is the same, except the running time is polynomial in the input size for any fixed $\epsilon$.

\section{Setting 1: Single-agent multi-actions}
\label{sec:model1}

In this section, we present and study the setting introduced by~\cite{DuettingEFK21}: a principal interacts with a single agent who may choose any subset of actions from a given set.
We refer to it throughout as the single-agent multi-action model.
We first describe the model and then present our results, summarized in~\Cref{tab:multi-action}.

\subsection{Model}
\label{sec:model1-model}

In the single-agent, multi-action model, a principal wishes to delegate a project to a single agent. The project has a binary outcome: it either succeeds or fails. We normalize the reward from success to $1$, and the reward from failure is $0$.

The agent has access to a collection of $n$ possible actions, $\actions = [n]$, and may select any subset of them. The performance of a chosen set $S \subseteq \actions$ is described by a success probability function $f:2^{\actions}\rightarrow [0,1]$, which assigns to each action set $S$ a probability $f(S)$ that the project succeeds. Because success yields a normalized reward of $1$, $f(S)$ also equals the expected reward of $S$. For this reason, we will often refer to $f$ as the (expected) reward function.

Each action $i \in [n]$ incurs a nonnegative cost $c_i \geq 0$. By default, we assume the cost function is additive, so that for any $S \subseteq \actions$, $c(S) := \sum_{i \in S} c_i$.

A linear contract specifies a payment $\alpha \in [0,1]$ from the principal to the agent in case of success, and $0$ otherwise. As we show below, for the binary-outcome setting, we can assume without loss of generality that the principal uses a linear contract. 
Given a contract $\alpha$, the agent's \emph{best response} is a set $S$ that maximizes the agent's expected utility, defined as the expected payment minus the incurred cost, i.e., 
$$
U_A(S \mid \alpha) := \alpha f(S)-c(S).
$$
When multiple sets maximize the agent's utility, the agent is assumed to break ties in favor of ones with higher reward $f(S)$ (this convention will become clear shortly). 

The principal's expected utility under contract $\alpha$ is 
$$
U_P(\alpha) := (1-\alpha)f(S_{\alpha}),
$$ 
where $S_{\alpha}$ denotes the agent's best response. In other words, the principal's utility equals the expected reward minus the expected payment for the agent's chosen set.
Note that the agent's tie-breaking rule in favor of higher $f$ values effectively breaks ties in favor of the principal's utility; this is a standard assumption in the literature.

A contract is said to be \emph{optimal} if it maximizes the principal's expected utility. 
While general contracts may assign arbitrary payments to outcomes, in the binary-outcome setting the optimal contract is always \emph{linear} --- paying $\alpha$ upon success and $0$ otherwise. Hence, restricting attention to linear contracts is without loss of generality~\citep{DuettingEFK21}.

The \emph{social welfare} of an action set $S$ is the difference between its reward and cost, and is given by $f(S)-\cost(S)$. 

\paragraph{Connection between an agent's best response and a demand query.}
An agent's best-response to a contract $\alpha$ is any set $S$ that maximizes $\alpha f(S) - \sum_{i \in S}\cost_i$. 
Clearly, this is equivalent to maximizing $f(S)- \sum_{i\in S} p_i$, with prices $p_i= c_i/\alpha$. 
Thus, computing a best response under $\alpha$ is exactly a \emph{demand query} at prices $p_i=c_i/\alpha$ (see Section~\ref{sec:oracle}).
A minor subtlety is that we assume the agent breaks ties in favor of sets with larger $f$ values. Thus, for the equivalence to hold, we adopt the same tie-breaking convention for demand queries.

\subsection{The geometry of linear contracts}
\label{sec:model1-linear}

In this section, we present a geometric representation of the agent's and principal's utilities. This perspective, introduced by~\cite{DuttingRT19} for the simple (non-combinatorial) principal-agent model, has proven especially useful in the combinatorial actions model of~\cite{DuettingEFK21}.

We begin with the agent’s perspective. Raising $\alpha$ increases the weight on the reward term relative to cost, so the agent’s utility is driven more by $f$ and less by $\cost$. Thus, as $\alpha$ grows, the agent is pushed toward action sets with higher reward, even when they come with higher cost. Formally, for each action set $\set$ we consider the utility curve $U_A(\set \mid \alpha)=\alpha f(\set)-\cost(\set)$ as a function of $\alpha$ (see Figure~\ref{fig:multi-action} (left)).
Thus, to determine which set is induced by a linear contract with parameter $\alpha$, pick the curve
that attains the maximum utility at that $\alpha$.
Equivalently, the agent’s best response traces the \emph{upper envelope} of this family of lines
(shown in bold in Figure~\ref{fig:multi-action} (left)).

If we list, from left to right, the sets that appear on the upper envelope (these are exactly the sets that can be induced by some~$\alpha$),
one can show that, with appropriate tie-breaking,  they are ordered by increasing reward $f(\set)$ (the slope) and, simultaneously, by increasing cost $\cost(\set)$ (negative height at $\alpha = 0$).

\begin{lemma}[Monotonicity~\citep{DuttingRT19,DuettingEFK21}]
\label{lem:utilityConvex}
    Let $S_\alpha, S_\beta \subseteq \actions$ be the agent's respective best responses for contracts $0 \le \alpha < \beta \le 1$, with ties broken in favor of the principal.
    Then, (1) $f(S_\alpha) < f(S_\beta)$, and (2) $c(S_\alpha) < c(S_\beta)$.
\end{lemma}

\newcommand{\sizeF}{\small}
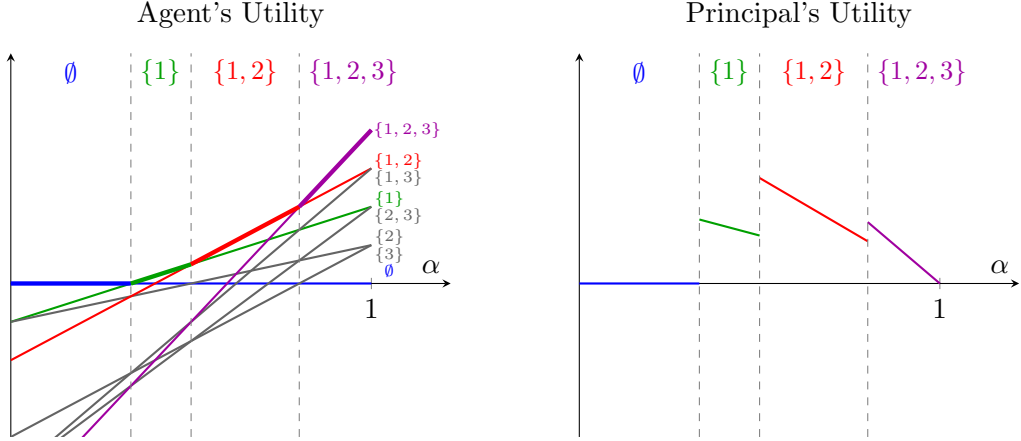
\begin{figure}
\centering
\begin{subfigure}[b]{0.45\textwidth}
\centering
\begin{tikzpicture}
\begin{axis}[
    width=\linewidth, height=0.9\linewidth,
    xmin=0, xmax=1.22,
    ymin=-0.4, ymax=0.6,
    axis x line=bottom,
    axis y line=left,
    xlabel={$\alpha$},
    xlabel style={at={(axis description cs:1.02,0)}, anchor=west},
    xtick={1},
    ytick=\empty,
    title={Agent's Utility},
    axis x line=middle,
]

%f({1}) = 0.3, f({2}) = 0.2, and f({3}) = 0.5. The action costs are c1 = c2 = 0.1, and c3 = 0.4.

% Additive f lines
\addplot[SetA, thick, domain=0:1] {0};          
\node[text=SetA] at (1.05,0.035) {\tiny $\emptyset$};
\addplot[SetB, thick, domain=0:1] {0.3*x - 0.1};
\node[text=SetB] at (1.05,0.22) {\tiny $\{1\}$};
\addplot[SetE, thick, domain=0:1] {0.2*x - 0.1};
\node[text=SetE] at (1.05,0.12) {\tiny $\{2\}$};
\addplot[SetE, thick, domain=0:1] {0.5*x - 0.4};
\node[text=SetE] at (1.05,0.075) {\tiny $\{3\}$};
\addplot[SetC, thick, domain=0:1] {0.5*x - 0.2};
\node[text=SetC] at (1.075,0.32) {\tiny $\{1,2\}$};
\addplot[SetE, thick, domain=0:1] {0.8*x - 0.5};
\node[text=SetE] at (1.075,0.275) {\tiny $\{1,3\}$};
\addplot[SetE, thick, domain=0:1] {0.7*x - 0.5};
\node[text=SetE] at (1.075,0.175) {\tiny $\{2,3\}$};
\addplot[SetD, thick, domain=0:1] {1.0*x - 0.6};
\node[text=SetD] at (1.1,0.4) {\tiny $\{1,2,3\}$};

\draw[dashed, gray] (0.333,-0.4) -- (0.333,0.6);
\node[text=SetA] at (0.166,0.55) {\sizeF $\emptyset$};
\draw[dashed, gray] (0.5,-0.4) -- (0.5,0.6);
\node[text=SetB] at (0.415,0.55) {\sizeF $\{1\}$};
\draw[dashed, gray] (0.8,-0.4) -- (0.8,0.6);
\node[text=SetC] at (0.65,0.55) {\sizeF $\{1,2\}$};
\node[text=SetD] at (0.95,0.55) {\sizeF $\{1,2,3\}$};

\newcommand{\thickness}{1.65pt}
% Thick lines
\addplot[SetA, line width=\thickness, domain=0:0.333] {0};
\addplot[SetB, line width=\thickness, domain=0.333:0.5] {0.3*x - 0.1};
\addplot[SetC, line width=\thickness, domain=0.5:0.8] {0.5*x - 0.2};
\addplot[SetD, line width=\thickness, domain=0.8:1] {1.0*x - 0.6};

\end{axis}
\end{tikzpicture}
\end{subfigure}
\begin{subfigure}[b]{0.45\textwidth}
\centering
\begin{tikzpicture}
\begin{axis}[
    width=\linewidth, height=0.9\linewidth,
    xmin=0, xmax=1.22,
    ymin=-0.4, ymax=0.6,
    axis x line=bottom,
    axis y line=left,
    xlabel={$\alpha$},
    xlabel style={at={(axis description cs:1.02,0)}, anchor=west},
    xtick={1},
    ytick=\empty,
    title={Principal's Utility},
    axis x line=middle,
]

\addplot[SetA, thick, domain=0:0.333] {0};
\addplot[SetB, thick, domain=0.333:0.5] {(1-x)*0.25};
\addplot[SetC, thick, domain=0.5:0.8] {(1-x)*0.55};
\addplot[SetD, thick, domain=0.8:1] {(1-x)*0.8};

\draw[dashed, gray] (0.333,-0.4) -- (0.333,0.6);
\node[text=SetA] at (0.166,0.55) {\sizeF $\emptyset$};
\draw[dashed, gray] (0.5,-0.4) -- (0.5,0.6);
\node[text=SetB] at (0.415,0.55) {\sizeF $\{1\}$};
\draw[dashed, gray] (0.8,-0.4) -- (0.8,0.6);
\node[text=SetC] at (0.65,0.55) {\sizeF $\{1,2\}$};
\node[text=SetD] at (0.95,0.55) {\sizeF $\{1,2,3\}$};

\end{axis}
\end{tikzpicture}
\end{subfigure}

\caption{The agent's and principal's utilities as a function of $\alpha$ for an additive reward function, with $f({1}) = 0.3$, $f({2}) = 0.2$, and $f({3}) = 0.5$, and action costs $c1 = c2 = 0.1, c3 = 0.4$.
Left: Upper envelope of the agent's utility. Right: The principal's utility for the agent's best response.
}
\label{fig:multi-action}
\end{figure}

\paragraph{Critical values.} 
A key role is played by the values of $\alpha$ where segments of the agent’s upper envelope intersect.
These intersection points are referred to as \emph{critical $\alpha$’s}, also called \emph{critical values} or \emph{breakpoints}. 
Interestingly, critical values are not only central in the computational analysis of contracts, but also in learning settings with unknown types: the number of critical values determines the pseudo-dimension of the contract class, and thereby its sample complexity (see~\Cref{sec:model1-sample}).

For a set $S$ on the upper envelope, the intersection with the preceding set $S'$ (or with the $x$-axis, for the first set) occurs at the value of $\alpha$ satisfying $U_A(S \mid \alpha) = U_A(S' \mid \alpha)$, that is, 
$
\alpha = \frac{f(S)-f(S')}{\cost(S)-\cost(S')}
$. 
At this $\alpha$ the agent is indifferent between $S$ and $S'$, and the choice is determined by the tie-breaking rule. Under the canonical tie-breaking rule (favoring the set that benefits the principal), the agent selects $S$ over $S'$.

Thus, the interval $[0,1]$ of feasible $\alpha$’s is partitioned into at most $2^n$ subintervals, one for each incentivizable set $S$. Each interval consists of all $\alpha$ values for which the agent chooses $S$.
(For notational convenience, we let all intervals be half-open. The last interval is actually closed.)

From the principal's perspective (see Figure~\ref{fig:multi-action}(right)), the goal is to reduce the agent's share $\alpha$ as much as possible while still incentivizing a set with high reward. Formally, for $\alpha \in [0,1]$, the principal’s utility is $(1-\alpha)f(\set_{\alpha})$, where $\set_{\alpha}$ is the agent's best response to $\alpha$.  
Within an interval $[\alpha,\beta)$ where the agent selects $S$, this utility traces a line that begins at height $(1-\alpha) f(S)$ and decreases with slope $-f(S)$. Consequently, the principal's best strategy for inducing $S$ is to choose the left endpoint $\alpha$ of its interval.

The next example illustrates this with a 3-action instance and additive reward function $f$.

\begin{example}[3-action instance with additive reward\
\citep{DuttingFT24-survey}]
\label{ex:additive} 
There are three actions $\actions=\{1,2,3\}$. The success probability function $f$ is additive, with $f(\{1\}) = 0.3, f(\{2\}) = 0.2$, and $f(\{3\})=0.5$. The action costs are $c_1=c_2=0.1$, and $c_3=0.4$.
Consider, for example, the contract $\alpha=0.5$. 
The agent's utility for taking action $1$ is $\alpha f(\{1\})-c_1=0.5\cdot 0.3-0.1=0.05$, for action $2$ it is $\alpha f(\{2\})-c_2=0.5\cdot 0.2-0.1=0$, and for action $3$ it is $\alpha f(\{3\})-c_3=0.5\cdot 0.5-0.4 = -0.15$. 
Therefore, among all singletons, 
action $1$ 
is best. However, the agent may be better off selecting more than a single action. The agent's utility for the set $\{1,2\}$ is $\alpha f(\{1,2\})-(c_1+c_2)=0.5\cdot 0.5-0.2=0.05$, for the set $\{1,3\}$ it is $\alpha f(\{1,3\})-(c_1+c_3)=0.5\cdot 0.8-0.5=-0.1$, for the set $\{2,3\}$ it is 
$\alpha f(\{2,3\})-(c_2+c_3)=0.5\cdot 0.7-0.5=-0.15$,
and for the set $\{1,2,3\}$ it is $\alpha f(\{1,2,3\})-(c_1+c_2+c_3)=0.5\cdot 1-0.6=-0.1$. 
Therefore at $\alpha = 0.5$ the agent is indifferent between $\{1\}$ and $\{1,2\}$ and breaks the tie in favor of the set $\{1,2\}$ (this point is the intersection of the green and red curves in Figure~\ref{fig:multi-action} (left)).
Below we provide more details about how the agent's best response changes as a function of $\alpha$, and how that affects the principal's choice of $\alpha$.
\end{example}

Figure~\ref{fig:multi-action}(left) illustrates the upper envelope for Example~\ref{ex:additive}.
Examining the diagram shows that the agent's best response evolves with $\alpha$: for small $\alpha$ the agent takes no action, then switches to action~1, then to the set $\{1,2\}$, and eventually to $\{1,2,3\}$ once $\alpha$ is large enough.
This pattern is not specific to this example. In any additive setting, action $i$ becomes part of the agent's best response precisely when $\alpha \geq c_i / f(\{i\})$, regardless of the other actions. Equivalently, this is the value of $\alpha$ that satisfies $\alpha f(S \cup \{i\}) - c(S \cup \{i\}) = \alpha f(S) - c(S)$.
As a result, there can be at most $n$ such critical $\alpha$’s.
The picture becomes considerably richer once we move beyond additive rewards.
As we will see, the number of these critical values play a central role in the optimal contract problem.

\subsection{Enumerating all critical values}
\label{sec:model1-ES}

\begin{figure}
\centering
\begin{minipage}{0.4\linewidth}
% \begin{algorithm}[H]
% \caption{{Enumerate CV}(L, R)}
% % \caption{\textsf{Enumerate CV} \protect\( (\lb,\rb) \protect\)}
% \label{alg:enumerate_CV}
% \begin{algorithmic}[1]
%     \IF{\(S_\rb = S_\lb\)}
%         \STATE return \(\emptyset\)
%     \ENDIF
%     \STATE \(\intersection \gets \frac{c(S_\rb) - c(S_\lb)}{f(S_\rb) - f(S_\lb)}\)
%     \IF{\(S_\intersection = S_\lb\)}
%         \STATE return \(\{\intersection\}\)
%     \ENDIF
%     \STATE \(C_1 \gets \texttt{CV}(\lb, \intersection)\)
%     \STATE \(C_2 \gets \texttt{CV}( \intersection, \rb)\)
%     \STATE return \(C_1 \cup C_2\)
% \end{algorithmic}
% \end{algorithm}
\begin{algorithm}[H]
\caption{\textsf{Enumerate CV} \protect\( (\lb,\rb) \protect\)}
\label{alg:enumerate_CV}
\begin{algorithmic}[1]
    \IF{$S_\rb = S_\lb$}
        \STATE \textbf{return} $\emptyset$
    \ENDIF
    \STATE $\intersection \gets \frac{c(S_\rb) - c(S_\lb)}{f(S_\rb) - f(S_\lb)}$
    \IF{$S_\intersection = S_\lb$}
        \STATE \textbf{return} $\{\intersection\}$
    \ENDIF
    \STATE $C_1 \gets \texttt{CV}(\lb, \intersection)$
    \STATE $C_2 \gets \texttt{CV}( \intersection, \rb)$
    \STATE \textbf{return} $C_1 \cup C_2$
\end{algorithmic}
\end{algorithm}
\end{minipage}
\hspace{4pt}
\begin{minipage}{0.4\linewidth}
\resizebox{.85\linewidth}{!}{
\begin{tikzpicture}
\begin{axis}[
width=\linewidth, height=0.9\linewidth,
ymin=0, ymax=0.25,
axis x line=bottom,
axis y line=left,
xlabel={$\alpha$},
xlabel style={at={(axis description cs:1.02,0)}, anchor=west},
ytick=\empty,
% title={Agent's Utility},
% axis x line=middle,
xtick={0.3,0.5,0.625,0.75,0.9},
xticklabels={$\lb$, , $\intersection$, , $\rb$},
title={Agent's Utility},
ticklabel style={font=\footnotesize}, % <-- added
every axis label/.append style={font=\footnotesize}, % <-- added
]

\addplot[domain=0.3:0.5, color=SetA, line width=1.5pt]{0.2*x-0.05};
\addplot[dotted, domain=0.3:0.9, color=SetA,line width=1pt]{0.2*x-0.05};
\addplot[dotted, domain=0.3:0.5, color=SetB,line width=1pt]{0.4*x-0.15};
\addplot[domain=0.5:0.75, color=SetB,line width=1.5pt]{0.4*x-0.15};
\addplot[dotted, domain=0.75:0.9, color=SetB,line width=1pt]{0.4*x-0.15};
\addplot[dotted, domain=0.3:0.75, color=SetC,line width=1pt]{0.6*x-0.3};
\addplot[domain=0.75:0.9, color=SetC,line width=1.5pt]{0.6*x-0.3};

\node[] at (0.37,0.06) {$S_\lb$};
\node[] at (0.59,0.13) {$S_\intersection$};
\node[] at (0.8,0.23) {$S_\rb$};

\draw[dashed, gray] (0.5,0) -- (0.5,0.25);
\draw[dashed, gray] (0.75,0) -- (0.75,0.25);
% \addplot[color=gray, dashed] coordinates{(0.5,0)(0.5,0.25)};
% \addplot[color=gray, dashed] coordinates{(0.75,0)(0.75,0.25)};

\end{axis}
\end{tikzpicture}
}
\end{minipage}
% \end{center}
\caption{A recursive algorithm for enumerating all critical values in a given interval.}
\label{fig:cv-alg}
\end{figure}

Let $f:2^{\actions} \to [0,1]$ and $c:2^{\actions} \to \reals_{\ge 0}$ denote the reward and cost set functions, respectively.
A \emph{best-response oracle} takes as input a contract $\alpha$ and returns a set of actions that maximizes the agent’s utility, $S \in \argmax_{T \subseteq \actions} \{\alpha f(T) - c(T)\}$, breaking ties in favor of higher $f$ values.

We next present an algorithm that, given access to a best-response oracle, enumerates all critical values in a specified interval, running in time linear in their number. This algorithm has appeared independently in several contexts. It was first introduced by 
% Eisner and Severance
~\cite{eisner1976mathematical} and later applied by~\cite{gusfield1980sensitivity} in sensitivity analysis, where it is often referred to as the \emph{Eisner–Severance} technique. In the setting of combinatorial contracts, it is used in~\citep{DuettingFG23} to establish the result stated below (see Algorithm~\ref{alg:enumerate_CV} in Figure~\ref{fig:cv-alg}).

\begin{proposition}[enumerating all critical values~\citep{DuettingFG23}]
\label{prop:optContractSuff}
    For any cost function $c:2^{\actions} \to \reals_{\ge 0}$ and reward function $f:2^{\actions} \to [0,1]$, given access to a best-response oracle, there exists an algorithm that enumerates all critical values in time linear in the number of critical values.
\end{proposition}

\begin{proof}[Proof sketch]
\Cref{alg:enumerate_CV} is recursive and its correctness follows from an inductive argument over the number of critical values.
Consider a segment $(\lb,\rb]\subseteq[0,1]$.
It is easy to verify that a segment has no critical values if and only if the agent's best response to contracts $\lb$ and $\rb$ is identical, i.e., $S_\lb =S_\rb$. This base case is described in Line 1 of \Cref{alg:enumerate_CV}.

Otherwise, let $\intersection = \frac{c(S_\rb)-c(S_\lb)}{f(S_\beta)-f(S_\lb)}$ be the intersection between the linear functions $U_A( S_\lb \mid \alpha)$ and $U_A(S_\rb \mid \alpha)$. 
If there is a single critical value in $(\lb,\rb]$, the agent's best response changes from $S_\lb$ to $S_\rb$ and thus the critical value is $\intersection$, as returned in Line 6.
Lastly, if there are two or more critical values in the segment. By the convexity and monotonicity of the agent's utility (\Cref{lem:utilityConvex}), the agent's best response for the contract $\intersection$ satisfies $f(S_\lb) < f(S_\intersection)$ and $c(S_\lb) < c(S_\intersection)$, which implies that there is at least one critical value in the segment $(\lb,\intersection]$.
Thus, applying the algorithm recursively for the two sub-segments $(\lb,\intersection]$ and $(\intersection,\rb]$, would yield, by the induction hypothesis, all critical values in $(\lb,\rb]$.
\end{proof}

As an implication of~\Cref{prop:optContractSuff}, we get the following theorem:

\begin{theorem}[Optimal contract algorithm~\citep{DuettingEFK21,DuettingFG23}]
\label{thm:alg-2conds}
In single-agent multi-action settings, any family of reward functions satisfying the following two conditions admits a polynomial-time algorithm for computing the optimal contract, with value queries:
\begin{enumerate}
    \item A polynomial-time algorithm for demand queries, and
    \item A polynomial bound on the number of critical values.
\end{enumerate}
\end{theorem}

\begin{proof}
Consider the following algorithm: for each critical value $\alpha$, compute the agent’s best response $S_{\alpha}$, and select an $\alpha$ that maximizes the principal’s expected utility $(1-\alpha)f(S_{\alpha})$.
Since the number of critical values is polynomial and each best-response query can be computed in polynomia time, the algorithm runs in polynomial time. Moreover, because the optimal contract must occur at a critical value, the algorithm indeed outputs an optimal contract.
\end{proof}

For instance, in the case of additive rewards, the two conditions in~\Cref{thm:alg-2conds} are satisfied, and hence the optimal contract can be computed efficiently. The situation becomes more challenging, however, for more complex reward functions $f$, as we explain below.

\subsection{Poly-time algorithm for gross substitutes (GS) rewards}
\label{sec:model1-GS}

Gross substitutes (GS) functions, defined in Section~\ref{sec:model}, form a strict subclass of submodular functions and encompass several interesting classes, such as additive and unit-demand. This class plays a major role both in economics, where it is the frontier for the existence of market equilibrium \citep{kelso1982job,gul1999walrasian}, and in computer science, where it admits a poly-time algorithm for social welfare maximization in combinatorial markets \citep{NisanSegal06}. In what follows, we show that GS functions likewise admit a polynomial-time algorithm for computing an optimal contract using value queries.

% \begin{wrapfigure}{r}{0.6\textwidth}
% \centering
% \input{figures/multi-action-gs}
% \caption{Upper envelope for Example~\ref{ex:gs}, with GS $f$}
% \label{fig:gs}
% \end{wrapfigure}

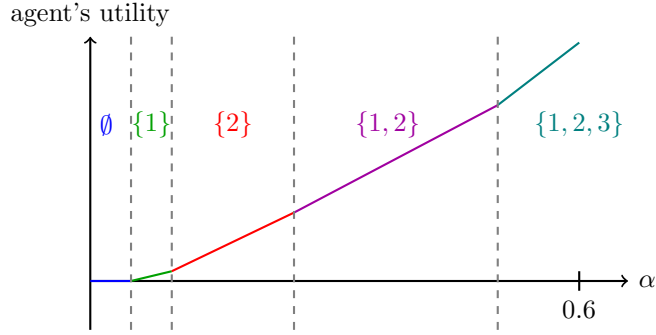
\begin{figure}
{\centering
\newcommand{\ymin}{-0.5}
\begin{tikzpicture}[scale=1.3]
% \begin{axis}
%     ymin=-0.5, ymax=3,
%     xmin=-1, xmax=6
% \end{axis}
    % \useasboundingbox (-1,-2) rectangle (6,3);
    % \clip (-1,\ymin) rectangle (6,3);
    %% scaling factors
    \def\z{2}
    \def\x{0.6}
    %% y-axis: 1 is 0.5
    \draw[thick,->] (0,\ymin) -- (0,2.5) node[above] {\small agent's utility};
    %% x-axis: 1 is 0.2
    \draw[thick,->] (0,0) -- (5.5,0) node[right] {\small $\alpha$};
    %% xticks / end of range
    \draw[thick,-] (5,0.1) -- (5,-0.1) node[below] {\small {\x}}; 
    %% {}: 0 
    %\draw[-,blue,thick] (0,0) -- (5,0) node[above,xshift=0.25cm,blue] {\tiny{$\emptyset$}};
    %% until alpha = 0.05
    \draw[-,SetA,thick] (0,0) -- (0.05*5/\x,0);
    %% {1}: alpha*1-0.05
    %\draw[-,green,thick] (0,-0.05*\z) -- (5,\x*\z*1-0.05*\z) node[above,xshift=0.325cm,green]{\tiny $\{1\}$};
    %% from alpha = 0.05 to alpha = 0.1
    \draw[-,SetB,thick] (0.05*5/\x,0.05*1*\z-0.05*\z) -- (0.1*5/\x,0.1*1*\z-0.05*\z);
    %% {2}: alpha*2-0.15
    %\draw[-,teal,thick] (0,-0.15*\z) -- (5,\x*\z*2-0.15*\z) node[above,xshift=0.325cm,gray]{\tiny $\{2\}$};
    %% from alpha = 0.1 to alpha = 0.25
    \draw[-,SetC,thick] (0.1*5/\x,0.1*2*\z-0.15*\z) -- (0.25*5/\x,0.25*2*\z-0.15*\z);
    %% {3}: alpha*1-0.5
    %\draw[-,gray,thick] (0,-0.5*\z) -- (5,\x*\z*1-0.5*\z) node[right]{\tiny $\{3\}$};;
    %% {1,2}: alpha*2.2-0.2
    %\draw[-,red,thick] (0,-0.2*\z) -- (5,\x*\z*2.2-0.2*\z) node[above,xshift=0.45cm,red]{\tiny $\{1,2\}$};
    %% from alpha = 0.25 to alpha = 0.5
    \draw[-,SetD,thick] (0.25*5/\x,0.25*2.2*\z-0.2*\z) -- (0.5*5/\x,0.5*2.2*\z-0.2*\z);
    %% {2,3}: alpha*3-0.65
    %\draw[-,gray,thick] (0,-0.65*\z) -- (5,\x*\z*3-0.65*\z) node[right]{\tiny $\{2,3\}$};;
    %% {1,3}: alpha*2-0.55
    %\draw[-,gray,thick] (0,-0.55*\z) -- (5,\x*\z*2-0.55*\z) node[right]{\tiny $\{1,3\}$};
    %% {1,2,3}: alpha*3.2-0.7
    %\draw[-,violet,thick] (0,-0.7*\z) -- (5,\x*\z*3.2-0.7*\z) node[right,violet]{\tiny$\{1,2,3\}$};
    %% from alpha = 0.5 to alpha = \x
    \draw[-,SetF,thick] (0.5*5/\x,0.5*3.2*\z-0.7*\z) -- (\x*5/\x,\x*3.2*\z-0.7*\z);
    %% intersection {} and {1} at alpha = 0.05
    \draw[gray,thick,dashed] (0.05*5/\x,\ymin) -- (0.05*5/\x,2.5);
    %% intersection {1} and {2} at alpha = 0.1
    \draw[gray,thick,dashed] (0.1*5/\x,\ymin) -- (0.1*5/\x,2.5);
    %% intersection {2} and {1,2} at alpha = 0.25
    \draw[gray,thick,dashed] (0.25*5/\x,\ymin) -- (0.25*5/\x,2.5);
    %% intersection {1,2} and {1,2,3} at alpha = 0.5
    \draw[gray,thick,dashed] (0.5*5/\x,\ymin) -- (0.5*5/\x,2.5);
    %% best'response regions
    \node[SetA] at (0.02*5/\x,1.6) {\small $\emptyset$};
    \node[SetB] at (0.075*5/\x,1.6) {\small $\{1\}$};
    \node[SetC] at (0.175*5/\x,1.6) {\small $\{2\}$};
    \node[SetD] at (0.365*5/\x,1.6) {\small $\{1,2\}$};
    \node[SetF] at (0.6*5/\x,1.6) {\small $\{1,2,3\}$};
\end{tikzpicture}
\caption{Upper envelope for Example~\ref{ex:gs}, with GS $f$}
\label{fig:gs}
}
\end{figure}

By~\Cref{thm:alg-2conds}, it suffices to establish two conditions: (1) a polynomial-time algorithm for computing demand queries, and (2) a polynomial bound on the number of critical values. The first condition holds for GS functions, as demand queries can be solved by a simple greedy algorithm (see~\Cref{sec:demand-gs}).
It remains to bound the number of critical values a GS function can admit. To build intuition for this problem, consider the simple GS instance with three actions described in the following example.

\begin{example}[3-action instance with GS reward] 
\label{ex:gs} 
There are three actions $\{1,2,3\}$.
The success probability function $f$ is as follows:
	$f(\emptyset)=0,~f(\{1\})=0.25,~f(\{2\}) = 0.5, f(\{3\}) = 0.25, f(\{1,2\})= 0.55,~f(\{1,3\}) = 0.5, f(\{2,3\}) = 0.75,$ and $f(\{1,2,3\}) = 0.8$.
The action costs are $c_1 =0.0125, c_2 =0.0375,$ and $c_3=0.125$.
Consider, for example, the contract $\alpha=0.5$ (this point is the intersection of the violet and teal curves in Figure~\ref{fig:gs}).
Given this contract, the agent's utility is maximized by set $\{1,2\}$ and set $\{1,2,3\}$.
Since the set $\{1,2,3\}$ yields a higher principal utility, 
the agent 
breaks the tie in favor of
this set. 
Below we provide more details about the transitions in the agent's best response, and how they differ from those in the additive case.
\end{example}

Figure~\ref{fig:gs} illustrates the upper envelope for the instance in Example~\ref{ex:gs}, where $f$ is a GS function. Tracing the envelope shows that the agent's best response evolves as follows: for small values of $\alpha$ the agent takes no action; as $\alpha$ increases, the chosen set is first $\{1\}$, then $\{2\}$ (replacing action 1), then $\{1,2\}$ (adding back action 1), and finally $\{1,2,3\}$. Unlike the clean progression in Example~\ref{ex:additive}, here action 1 is included for some range of $\alpha$, later dropped, and then chosen again at a larger $\alpha$. Thus, the number of critical values can exceed $n$, and the sequence of transitions along the $\alpha$ axis may be more intricate than in the additive case. Despite this complexity, \cite{DuettingEFK21} prove the following upper bound on the number of critical values for GS functions.

\begin{lemma}
\label{lem:poly-CV-GS}
Every GS reward function $f$ admits at most $O(n^2)$ critical values.
\end{lemma}

% We now sketch the proof of~\Cref{lem:poly-CV-GS}. 
% Our presentation deviates from the original proof in~\cite{DuettingEFK21}: it combines their techniques with with arguments from~\cite{demand_working_paper}, which establishes a connection to the geometric approach developed in~\cite{demand_types} that builds on the notion of \emph{demand types}.

We now sketch the proof of~\Cref{lem:poly-CV-GS}.
Our presentation deviates from the original proof in~\citep{DuettingEFK21}: it combines their techniques with arguments from~\cite{demand_working_paper}, which link the analysis to the geometric approach of~\cite{demand_types}, based on the notion of \emph{demand types}.

\begin{proof} [Proof sketch]
    Recall that $D_f(p) = \argmax_{S\subseteq A} \{f(S)-\sum_{i\in S} p_i\}$ is the set of all bundles that maximize demand at prices $p$ (see~\Cref{sec:oracle}).
\cite{demand_types} propose a geometric view of demand. They show that 
for each price vector $p \in \reals^n$, $D_f(p)$ is a single set, except on a measure-zero set. This exceptional set, where demand is not unique, is called the \emph{LIP (Locus of Indifference Prices)}; it consists of $(n-1)$-dimensional linear pieces known as \emph{facets}. By continuity of $f(S)-\sum_{i\in S} p_i$, removing the LIP partitions $\reals^n$ into $n$-dimensional open regions in which $|D_f(p)|=1$; these are the \emph{UDRs (Unique Demand Regions)}. \Cref{fig:facet_piercing} illustrates this in a 2-dimensional setting (for three different functions $f$): the colored lines are facets, the LIP is their union, and each UDR is labeled by its unique demanded set. (The dashed line will be explained shortly.)

A key lemma of~\cite{demand_types} states that a function is GS if and only if, when crossing a facet between two adjacent UDRs, the demanded set changes in exactly one of two ways: either a single element is added or removed, or one element is exchanged for another. This holds in the two leftmost diagrams in~\Cref{fig:facet_piercing}, but not in the rightmost panel, where moving across the orange facet corresponds to adding or removing two elements simultaneously.

\cite{demand_working_paper} study this perspective in the context of critical values in the multi-action setting. Recall that the agent's best response can be expressed as a demand query: maximizing the agent's utility is equivalent to finding a set in $D_f(\vec{c}/\alpha)$. In the geometric framework of \cite{demand_types}, this corresponds to a line in $\mathbb{R}^n$, illustrated as the dashed line in \Cref{fig:facet_piercing}. They show that, without loss of generality, the linear contract line ``interacts nicely'' with the LIP, in the sense that it moves between UDRs by crossing one facet at a time. A critical point arises exactly when the contract line intersects a facet of the LIP, causing the agent's best response to shift from the unique demand of one UDR to that of the next. Consequently, when $f$ is gross substitutes, the transition is always of one of two forms: either an action is added or removed, or one action is replaced by another. Moreover, since the cost of the agent's best response increases with $\alpha$, raising $\alpha$ corresponds either to adding an action, or to substituting a cheaper action with a more costly one.
    
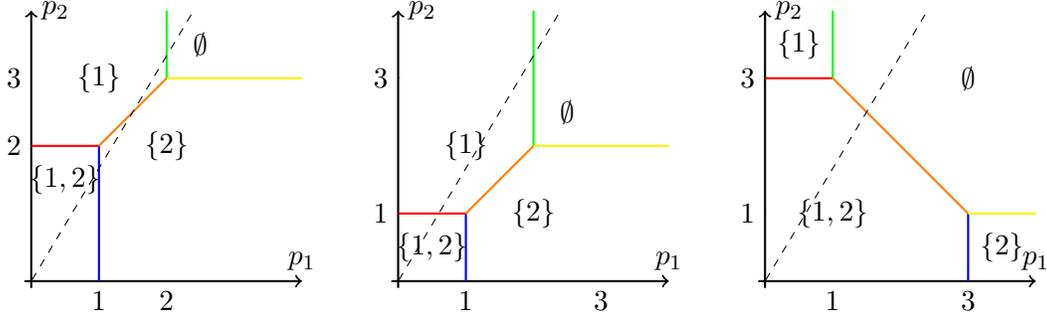
\begin{figure}
\begin{center}
\begin{tikzpicture}[scale=0.9]
    \draw[thick, ->] (-0.1, 0) -- (4, 0) node[above] {\normalsize $p_1$};
    \draw[thick, ->] (0, -0.1) -- (0, 4) node[right] {\normalsize $p_2$};
    
    \draw[blue,thick] (1, 0) -- (1, 2);
    \draw[red,thick] (0, 2) -- (1, 2);
    \draw[orange,thick] (1, 2) -- (2, 3);
    \draw[green,thick] (2, 3) -- (2, 4);
    \draw[yellow,thick] (2, 3) -- (4, 3);
    
    \node at (0.5, 1.5) {$\{1,2\}$};
    \node at (2.5, 3.5) {$\emptyset$};
    \node at (1, 3) {$\{1\}$};
    \node at (2, 2) {$\{2\}$};

    \draw (1,0) -- (1,0) node [below] {$1$};
    \draw (2,0) -- (2,0) node [below] {$2$};
    \draw (0,2) -- (0,2) node [left] {$2$};
    \draw (0,3) -- (0,3) node [left] {$3$};

    \draw[,dashed] (0.012,0.02) -- (2.4,4);
\end{tikzpicture}
% f(a) = 2, f(b)=3, f(a,b)=4
\quad
\begin{tikzpicture}[scale=0.9]
    \draw[thick, ->] (-0.1, 0) -- (4, 0) node[above] {\normalsize $p_1$};
    \draw[thick, ->] (0, -0.1) -- (0, 4) node[right] {\normalsize $p_2$};
    
    \draw[blue,thick] (1, 0) -- (1, 1);
    \draw[red,thick] (0, 1) -- (1, 1);
    \draw[orange,thick] (1, 1) -- (2,2);
    \draw[green,thick] (2, 2) -- (2, 4);
    \draw[yellow,thick] (2, 2) -- (4, 2);

    \node at (0.5, 0.5) {$\{1,2\}$};
    \node at (2.5, 2.5) {$\emptyset$};
    \node at (1, 2) {$\{1\}$};
    \node at (2, 1) {$\{2\}$};

    \draw (1,0) -- (1,0) node [below] {$1$};
    \draw (3,0) -- (3,0) node [below] {$3$};
    \draw (0,1) -- (0,1) node [left] {$1$};
    \draw (0,3) -- (0,3) node [left] {$3$};

    \draw[,dashed] (0.012,0.02) -- (2.4,4);
\end{tikzpicture}
\quad
\begin{tikzpicture}[scale=0.9]
    \draw[thick, ->] (-0.1, 0) -- (4, 0) node[above] {\normalsize $p_1$};
    \draw[thick, ->] (0, -0.1) -- (0, 4) node[right] {\normalsize $p_2$};
    \draw[blue,thick] (3, 0) -- (3, 1);
    \draw[red,thick] (0, 3) -- (1, 3);
    \draw[orange,thick] (1, 3) -- (3,1);
    \draw[green,thick] (1, 3) -- (1, 4);
    \draw[yellow,thick] (3, 1) -- (4, 1);

    \node at (1, 1) {$\{1,2\}$};
    \node at (3, 3) {$\emptyset$};
    \node at (0.5, 3.5) {$\{1\}$};
    \node at (3.5, 0.5) {$\{2\}$};

    \draw (1,0) -- (1,0) node [below] {$1$};
    \draw (3,0) -- (3,0) node [below] {$3$};
    \draw (0,1) -- (0,1) node [left] {$1$};
    \draw (0,3) -- (0,3) node [left] {$3$};
    \draw[,dashed] (0.012,0.02) -- (2.4,4);
\end{tikzpicture}
\end{center}
\caption{Three 2-action instances.  
% examples of the linear contract line. 
In all examples $c(1)=0.012, c(2)=0.02$, and they differ only in $f$. Left: $f(\emptyset) = 0, f(\{1\}) = 2, f(\{2\})=3, f(\{1,2\})=4$.
Middle: $f(\emptyset) = 0, f(\{1\}) = 2, f(\{2\})=2, f(\{1,2\})=4$. Right: $f(\emptyset) = 0, f(\{1\}) = 1, f(\{2\})=1, f(\{1,2\})=4$.
}
\label{fig:facet_piercing}
\end{figure}
% f(a) = 3, f(b)=3, f(a,b)=4

The $O(n^2)$ upper bound then follows via a simple potential-function argument. Rank each action by its cost (1 for the cheapest, $n$ for the most expensive), and define the potential of a set as the sum of its members' ranks. This potential is bounded above by $\sum_{i=1}^n i = O(n^2)$. Since the potential of the best response is an integer that increases monotonically with $\alpha$ (due to the restricted forms of transitions identified above), the number of transitions is at most $O(n^2)$.
\end{proof}

Two remarks are in order. 
First, the $O(n^2)$ bound is tight; indeed, there exists a GS reward function $f$ admitting 
$\Theta(n^2)$ critical values. 
Second, for the broader class of submodular rewards, the condition of having only polynomially many critical values fails. Indeed, \cite{DuettingEFK21} construct a submodular reward function $f$ with \emph{all} $2^n-1$ sets of actions appearing along the upper envelope.

Combining~\Cref{thm:alg-2conds}, \Cref{thm:gs-greedy}, and \Cref{lem:poly-CV-GS} yields the main result for GS rewards:

\begin{theorem}[Poly-time algorithm for GS rewards, with value oracle~\citep{DuettingEFK21}]
    \label{thm:GS-alg}
    In single-agent multi-action settings, with GS reward functions, the optimal contract can be computed in time polynomial in $n$, given access to a value oracle.
\end{theorem}

\subsection{Poly-time algorithm for ultra rewards}
\label{sec:ultra}

Recall that GS serves as a tractability frontier for several fundamental problems. A natural question, raised in~\citep{DuettingEFK21}, is whether GS also marks the frontier for efficient algorithms for the optimal contract problem. An initial step beyond GS was taken by~\cite{DuettingFG23}, who established tractability in the \emph{supermodular} regime (see~\Cref{sec:supermod}). More recently,~\cite{FeldmanY25} showed that tractability extends to the class of \emph{ultra} rewards --- a strict superclass of GS. For clarity of exposition, we present the ultra result first, since its analysis is more similar to that of GS, while the supermodular case requires a different line of reasoning.

Ultra functions admit multiple equivalent definitions; we present the one we find most instructive --- the characterization via the \emph{well-layered} property --- as it aligns with a greedy procedure for computing demand.

\begin{definition}\citep{dress1995well,Lehman2017}
\label{def:ultra}
A set function $f:2^{[n]}\rightarrow \reals_{\geq 0}$ is said to be \emph{ultra} if it satisfies the \emph{well-layered} property, defined as follows: for every additive function $p$,  the sequence of sets $S_0,S_1,\ldots, S_n$ obtained by the following greedy procedure:
    $S_0=\emptyset \text{  and  }S_{i}=S_{i-1} \cup x_i, \text{  for  } x_i \in \argmax_{x\notin S_{i-1}} \{f(x \mid S_{i-1})-p(x)\}$, with ties broken arbitrarily, 
    are such that $S_i \in \argmax_{S,|S|=i} \{ f(S)-p(S)\}$ for every $i$.
\end{definition}

This definition implies that demand queries for ultra functions can be computed efficiently.

\begin{lemma}\citep{dress1995well,Lehman2017,FeldmanY25}
\label{lemma:contract-demand-query-ultra-add}
For any ultra function $f$, and price vector $p$, a demand query for $(f,p)$ can be solved using Algorithm \textsf{GreedyUltra} (\Cref{alg:ultra-demand-query}).
\end{lemma}

As in the GS case, to ensure the equivalence between demand and best-response queries, we must handle the subtlety of tie-breaking in favor of the principal's utility. Accordingly, in both instances where ties may arise (lines 3 and 6 of~\Cref{alg:ultra-demand-query}), we break them by choosing the set with the higher $f$ value.

Note that the greedy procedure in~\cref{alg:ultra-demand-query} closely resembles that in~\cref{alg:gs-demand-query}. The key difference is that, unlike the GS case, the algorithm continues even after encountering a negative marginal. This is because ultra functions are \emph{not} submodular, and the marginal contribution of an action may increase even after turning negative. In fact, within the class of monotone set functions, GS is precisely the intersection of submodular and ultra functions~\citep{FeldmanY25}.

We now present the main theorem for ultra rewards.

\begin{theorem}[Poly-time algorithm for ultra rewards~\citep{FeldmanY25,demand_working_paper}]
    \label{thm:ultra-alg}
    In single-agent multi-action settings, with ultra reward functions, the optimal contract can be computed in poly time in $n$, given access to a value oracle.
\end{theorem}

By~\Cref{thm:alg-2conds}, and~\Cref{lemma:contract-demand-query-ultra-add}, to prove~\Cref{thm:ultra-alg}, it remains to bound the number of critical values.

Recall from the proof of~\Cref{lem:poly-CV-GS} that, w.l.o.g., critical values arise only when the linear-contract line crosses a facet of the LIP.
As shown in~\citep{demand_working_paper} (echoing the argument in~\citep{FeldmanY25}), for ultra functions crossing a facet between adjacent UDRs changes the demand by either (i) adding or removing one or more elements, or (ii) exchanging one element for another.
This strictly generalizes the GS case, which allows only single-element additions or removals.
Using the same potential-function argument as in the GS proof, the potential increases at every critical value and remains bounded by $O(n^2)$, thereby extending the GS bound to ultra rewards.

\subsection{Poly-time algorithm for supermodular rewards}
\label{sec:supermod}

While most positive results so far have focused on reward functions within the complement-free hierarchy of \cite{LehmannLN06}, in this section we present a result from \citep{DuettingEFK23,VuongDPP23} showing that instances with \emph{supermodular} rewards admit a polynomial-time algorithm for computing the optimal contract. This holds even when the cost function $c$ is submodular.

\begin{theorem}[Poly-time algorithm for supermodular rewards~\citep{DuettingEFK23,VuongDPP23}]
    \label{thm:supermod}
    In single-agent multi-action settings, with supermodular reward functions and submodular cost functions, the optimal contract can be computed in time polynomial in $n$, given access to a value oracle.
\end{theorem}

\begin{proof}[Proof sketch]
By~\Cref{thm:alg-2conds}, it suffices to show that such instances admit an efficient algorithm for computing a demand query (or equivalently, the agent's best response), and polynomially-many critical values.

We start with the former condition. If $c$ is submodular, and $f$ is supermodular, then the agent's utility, $\alpha f(S) - c(S)$, is supermodular. 
Finding a set $S$ that maximizes the agent's utility is equivalent to finding a set which minimizes the negation of the agent's utility, which is submodular. 
It is well known that there exists a polynomial time algorithm for submodular function minimization \citep{IwataFF01}.  
Additionally, it is easy to observe that the set of minimizers of a submodular function is closed under intersection and union. Thus, to handle tie breaking properly, it is enough to adapt the algorithm of \cite{IwataFF01} to return the maximal minimizer (which has the largest $f$ value, by monotonicity).
This can be done with no extra asymptotic cost as shown in \citep{nagano_submodular_min}.
Thus, a best-response oracle for the agent can be implemented efficiently.

For the second condition, the following lemma implies that the number of critical values is at most $n$.

\begin{lemma}[\citep{DuettingEFK23}]
\label{lem:SuperModContainment}
    For any supermodular reward function $f$, any submodular cost function $c$, and any two contracts $\alpha < \alpha'$ and corresponding sets in the agent's best response $S_\alpha$, $S_{\alpha'}$, it holds that $S_\alpha \subseteq S_{\alpha'}$.
\end{lemma}

\begin{proof}
    If $S_\alpha = S_{\alpha'}$ the claim obviously hold.
    Otherwise, assume that $S_{\alpha'}$ is a maximal best-response for contract $\alpha'$ (this is consistent with our tie-breaking assumption). Let $R = S_\alpha \setminus S_{\alpha'}$, and assume for the sake of contradiction that $R \ne \emptyset$. 
    We will show that the marginal contribution of $R$ to $S_{\alpha'}$ at $\alpha'$ is $\geq 0$, contradicting the maximality of $S_{\alpha'}$. Indeed, the marginal contribution of $R$ to $S_{\alpha'}$ at $\alpha'$ is given by 
    \begin{align*}
    \alpha' f(R \mid S_{\alpha'}) - c(R \mid S_{\alpha'}) 
    &\ge \alpha' f(R \mid S_\alpha \cap S_{\alpha'}) - c(R \mid S_\alpha \cap S_{\alpha'}) \\
        % && \text{[by supermodularity of $f$ and submodularity of $c$]} \\
    &\ge \alpha f(R \mid S_\alpha \cap S_{\alpha'}) - c(R \mid S_\alpha \cap S_{\alpha'}),
        % && \text{[by monotonicity of $f$]} 
\end{align*}
where the first inequality holds by the supermodularity of $f$ and submodularity of $c$, and the second inequality follows by the monotonicity of $f$.
The last term is precisely the marginal contribution of $R$ to $S_\alpha$ at $\alpha$, which is non-negative by the optimality of $S_{\alpha}$ at $\alpha$.
This concludes the proof of the lemma.
\end{proof}
Having verified both conditions of~\Cref{thm:alg-2conds}, the theorem follows.
\end{proof}

\subsection{Hardness results for submodular and XOS rewards}
\label{sec:model1-hardness}

\subsubsection{Hardness of optimal contract.}
\label{sec:hardness-submod-demand}

As noted above, some submodular functions admit exponentially many critical values. 
But this alone does not imply computational hardness; it merely rules out the brute-force approach of enumerating all critical values and selecting the best. 
Nevertheless, \cite{DuettingEFK21} show that the optimal contract problem for submodular rewards is \textsf{NP}-hard with value queries alone. 
The proof proceeds via a reduction from the \textsf{subset sum} problem. In the induced instance, only two candidate values of $\alpha$ can yield the optimal contract, but determining which of these is superior requires computing the agent's best response at these $\alpha$'s, which is hard. Thus, the hardness here stems from the hardness of solving a demand query for submodular functions.

A fundamental open question was whether the problem becomes tractable under demand-oracle access. \cite{DuettingFGR26} answer this in the negative, showing that computing an optimal contract for submodular $f$ requires exponentially many demand queries, establishing the intrinsic hardness of the problem under submodular rewards.

\begin{theorem}[Hardness of submodular rewards with demand oracle~\citep{DuettingEFK24}]
\label{thm:hardness-submodular-demand}
In single-agent multi-action settings, with submodular rewards, any algorithm that computes the optimal contract requires exponentially many demand queries to $f$.
\end{theorem}

\begin{proof}[Proof sketch]
The demand query hardness is obtained by constructing an instance that satisfies the following two properties:
(i) \emph{Equal revenue}: There exist exponentially many distinct contracts, each incentivizing the agent to take a different set of actions, yet all yielding the same utility for the principal.
% (see~\Cref{fig:eqrev}).
(ii) \emph{Sparse demand}: For any price vector $p$, the number of sets $S$ that maximize $f(S) - \sum_{i \in S} p_i$, up to a small additive factor, is polynomial.

The equal revenue property enables a standard ``hide a special set'' argument, which is the basis for value-query hardness. Combined with the sparse demand property, this argument extends to the demand-oracle model, yielding the exponential lower bound. 
Intuitively, the sparse demand property ensures that a demand query gives information only on poly-many sets, so one needs exponentially-many demand queries to find the hidden set.
\end{proof}

Beyond query complexity, \cite{DuettingEFK24} also introduce a communication-complexity framework, where the contract instance is split across two parties. They provide strong lower bounds, under different combinations of submodular/supermodular $f$ and $c$, further reinforcing the intrinsic hardness of the optimal contract problem. A full treatment of this perspective is outside the scope of this survey; see \citep{DuettingEFK24} for details.

\subsubsection{Hardness of approximation.}
\label{sec:hardness-approx-submod}
In this section we present hardness of approximation results for submodular and XOS rewards.
The following theorem, established in~\citep{EzraFS24}, shows that obtaining any constant-factor approximation for submodular rewards is hard.

\begin{theorem}[Submodular rewards, hardness of approximation~\citep{EzraFS24}]
\label{thm:submod-approx-hard}
    In single-agent multi-action settings, 
    with submodular rewards, no polynomial-time algorithm with value oracle access can approximate the optimal contract to within any constant factor, unless \textsf{P} $=$ \textsf{NP}.
\end{theorem}

\begin{proof}[Proof sketch]
The hardness result builds on an \textsf{NP}-hard promise problem for normalized unweighted coverage functions (a subclass of submodular functions), which is a generalization of a result by \cite{feige98}. The problem is to distinguish between a ``good'' case, where a small set $S$ achieves $f(S)=1$, and a ``bad'' case, where only significantly larger sets can approximate $f$.

In the contract-design context, this promise problem can be reduced to the principal's optimization problem. 
A subtlety arises, however: the reduction above shows hardness of computing the optimal utility, but since computing the utility of a given contract is not always feasible, the hardness of computing the optimal utility does not necessarily imply hardness of finding the optimal contract.
To address this, the reduction is modified so that approximately optimal contracts are separated. The key step is the introduction of an additional action that can only be incentivized at a high payment. This action effectively separates the two cases: in the good case, optimal contracts avoid it (yielding a low payment), while in the bad case, optimal contracts must include it (forcing a high payment). This ensures that finding an approximately optimal contract is computationally hard.
\end{proof}

For the broader class of XOS rewards, a stronger hardness applies:

\begin{theorem}[XOS rewards, hardness of approximation~\citep{EzraFS24}]
\label{thm:xos-approx-hard}
    In single-agent multi-action settings, 
    with XOS rewards, for any $\epsilon > 0$, no polynomial-time algorithm with value query access can approximate the optimal contract to within a factor of $n^{1/2 - \epsilon}$, unless \textsf{P} $=$ \textsf{NP}.
\end{theorem}

This hardness is established via a reduction from the well-studied problem of approximating the size of the largest clique in a graph. Classic results of~\cite{Hastad_1999} and~\cite{zuckerman2006} show that, unless \textsf{P}$=$\textsf{NP}, no polynomial-time algorithm can approximate the clique number $\omega(G)$ within a factor of $n^{1-\epsilon}$ for any $\epsilon>0$. The reduction links the two problems by showing that an algorithm achieving a $\beta$-approximation for the optimal contract would translate into an approximation of $\omega(G)$ within a factor of $\beta^2/4$. This connection yields the inapproximability bound of $n^{1/2 - \epsilon}$ for contracts with XOS rewards. A full proof of this reduction is given in~\citep{EzraFS24}.

\subsection{FPTAS for arbitrary monotone rewards with demand oracle}
\label{sec:fptas-demand}

The results above establish hardness of approximation for natural reward classes such as submodular and XOS functions. Strikingly, this barrier disappears once demand-oracle access is allowed. The following theorem, due to~\cite{DuettingEFK24,DuettingEFK21-sicomp}, shows that every monotone success probability function $f$ admits an FPTAS.\footnote{A weakly polynomial-time FPTAS was given earlier in \citep{DuettingEFK21}.}

\begin{theorem}[FPTAS with demand oracle~\citep{DuettingEFK24,DuettingEFK21-sicomp}]
\label{thm:FPTAS-monotone}
In single-agent multi-action settings, with monotone rewards, there exists an algorithm that gives a $(1-\epsilon)$-approximation to the optimal principal utility with $O(n^2/\epsilon)$ many value and demand queries to the reward function $f$.
\end{theorem}

Two remarks are in order. First, this result is best possible. Indeed, computing an optimal contract is \textsf{NP}-hard for submodular $f$~\citep{DuettingEFK21}, and requires exponentially many demand oracle calls (see~\Cref{thm:hardness-submodular-demand}).
Second,~\Cref{thm:FPTAS-monotone} remains true even when the cost function $c$ is subadditive, assuming access to a best-response oracle (i.e., given, a contract $\alpha$, return the set of actions that maximizes the agent's utility)~\citep{DuettingEFK21-sicomp}.

A key challenge in designing the FPTAS is that the optimal contract may lie arbitrarily close to $1$. Indeed, if it were bounded away from $1$, a sufficiently fine discretization of the contract space would suffice to approximate the optimum. 
To address this, \cite{DuettingEFK24} leverage a reduction to the non-combinatorial contract model of \cite{DuttingRT19}, which shows that the gap between the maximum \emph{welfare} and the principal's utility under the optimal contract is upper bounded by the number of ``actions'' --- $2^n$ in the combinatorial case.

\subsection{The pseudo-dimension of contracts}
\label{sec:model1-sample}

In a learning-based model of contract design~\citep{DuttingFPS25}, the agent's type, capturing both costs and rewards of actions, is unknown and can only be accessed through i.i.d. samples from an underlying (and a priori unknown) distribution. This naturally raises the question of sample complexity: how many samples are required to guarantee that the learned contract achieves near-optimal utility with high probability? To address this, \cite{DuttingFPS25} connects contract design with statistical learning theory via the notion of \emph{pseudo-dimension} --- a combinatorial measure of the expressive power of a class of contracts \citep{pollard1984convergence}.

A key structural insight is the role of critical values in determining the pseudo-dimension of linear contracts.

\begin{theorem}[Pseudo-dimension of contracts {\citep{DuttingFPS25}}]
In single-agent multi-action settings, the pseudo-dimension of linear contracts is $\Theta(\log N)$, where $N$ is the number of critical values.
\end{theorem}

As a consequence, the \emph{sample complexity} of learning contracts is polynomial in $\log N$. This represents an exponential improvement over naive bounds that depend on the number of actions, and shows that contracts can be learned efficiently from relatively few samples in well-studied scenarios. Thus, the number of critical values emerges as a unifying parameter that governs both the computational tractability and the statistical learnability of contract design.

\begin{table}[t]
\fontsize{10pt}{12pt}
\selectfont
\scalebox{0.95}{
\begin{tabular}{|c|cc|cc|}
\hline
\rowcolor[HTML]{C0C0C0} 
\textbf{\begin{tabular}[c]{@{}c@{}}Multiple\\ actions\end{tabular}} & \multicolumn{2}{c|}{\cellcolor[HTML]{C0C0C0}\textbf{Value Oracle}}                                                                                                                                                                       & \multicolumn{2}{c|}{\cellcolor[HTML]{C0C0C0}\textbf{Value and Demand Oracle}}                                                                                                              \\ \hline
\rowcolor[HTML]{C0C0C0} 
\multicolumn{1}{|l|}{\cellcolor[HTML]{C0C0C0}}                      & \multicolumn{1}{c|}{\cellcolor[HTML]{C0C0C0}\textbf{\begin{tabular}[c]{@{}c@{}}Upper bound \\ (pos)\end{tabular}}} & \textbf{\begin{tabular}[c]{@{}c@{}}Lower bound\\ (neg)\end{tabular}}                                                & \multicolumn{1}{c|}{\cellcolor[HTML]{C0C0C0}\textbf{\begin{tabular}[c]{@{}c@{}}Upper bound \\ (pos)\end{tabular}}}  & \textbf{\begin{tabular}[c]{@{}c@{}}Lower bound\\ (neg)\end{tabular}} \\ \hline
\cellcolor[HTML]{C0C0C0}\textbf{GS}
               
& \multicolumn{1}{c|}{\cellcolor[HTML]{FFFFC7}\begin{tabular}[c]{@{}c@{}}1\\ \cite{DuettingEFK21}\\
\cite{FeldmanY25}\end{tabular}}                    &  {\cellcolor[HTML]{EFEFEF} 1}                                                                                                                   & \multicolumn{1}{c|}{\cellcolor[HTML]{EFEFEF}1}                                                                      &  {\cellcolor[HTML]{EFEFEF} 1}                                                                     \\ \hline
\cellcolor[HTML]{C0C0C0}\begin{tabular}[c]{@{}c@{}}\textbf{Sub-}\\\textbf{modular}\end{tabular}                          & \multicolumn{1}{c|}{\cellcolor[HTML]{FFFFFF}}                                                                      & \cellcolor[HTML]{FFFFC7}\begin{tabular}[c]{@{}c@{}}No constant\\ approx\\ (if \textsf{P}$\neq$\textsf{NP})\\ \footnotesize{\cite{EzraFS24}}\end{tabular}       & \multicolumn{1}{c|}{\cellcolor[HTML]{EFEFEF}
FPTAS}
& 
\multicolumn{1}{c|}{\cellcolor[HTML]{FFFFC7}
\begin{tabular}[c]{@{}c@{}}$>1$\\ 
\footnotesize{\cite{DuettingEFK21}}\\
\footnotesize{\cite{DuettingFGR26}}
\end{tabular}}
\\ 
\hline
\cellcolor[HTML]{C0C0C0}\textbf{XOS}                                & \multicolumn{1}{c|}{\cellcolor[HTML]{FFFFFF}}                                                                      & \cellcolor[HTML]{FFFFC7}\begin{tabular}[c]{@{}c@{}}No better\\ than $n^{1/2-\epsilon}$\\ (if \textsf{P}$\neq$\textsf{NP})\\ \footnotesize{\cite{EzraFS24}}\end{tabular} & \multicolumn{1}{c|}{\cellcolor[HTML]{EFEFEF}
FPTAS}
& 
\multicolumn{1}{c|}{\cellcolor[HTML]{EFEFEF}
$>1$}\\ \hline
\cellcolor[HTML]{C0C0C0} 
\begin{tabular}[c]{@{}c@{}}\textbf{Sub-}\\\textbf{additive}\end{tabular}  
& \multicolumn{1}{c|}{\cellcolor[HTML]{FFFFFF}\textbf{}}                                                             & \cellcolor[HTML]{EFEFEF}\begin{tabular}[c]{@{}c@{}}No better\\ than $n^{1/2-\epsilon}$\end{tabular}                       & 
\multicolumn{1}{c|}{\cellcolor[HTML]{FFFFC7}
\begin{tabular}[c]{@{}c@{}}FPTAS\\ \footnotesize{\cite{DuettingEFK21}}\\
\footnotesize{\cite{DuettingEFK24}}
\end{tabular}} &
\multicolumn{1}{c|}{\cellcolor[HTML]{EFEFEF}
$>1$}\\ \hline \hline
\cellcolor[HTML]{C0C0C0}\cellcolor[HTML]{C0C0C0}\begin{tabular}[c]{@{}c@{}}\textbf{Super-}\\\textbf{modular}\end{tabular}                       & 
\multicolumn{1}{c|}{\cellcolor[HTML]{FFFFC7}\begin{tabular}[c]{@{}c@{}}1\\ \footnotesize{\citet{DuettingFG23}}\\ \footnotesize{\citet{VuongDPP23}}\end{tabular}}
& {\cellcolor[HTML]{EFEFEF} 1}                                                                                                                   & \multicolumn{1}{c|}{\cellcolor[HTML]{EFEFEF}1}                                                                      &  {\cellcolor[HTML]{EFEFEF} 1}                                                                      \\ \hline
\end{tabular}
}
\caption{Summary of approximation guarantees in the single-agent multi-action setting.
Left: value-oracle access. Right: value- and demand-oracle access.
For each oracle model, we list known upper bounds (positive) and lower bounds (hardness) across reward-function classes.
Yellow entries indicate main results; gray entries follow by closure (positive extend up/right; negative down/left).
}
\label{tab:multi-action}
\end{table}

\subsection{Open problems and future directions}
\label{sec:model1-open}

The single-agent multi-action model has inspired a flurry of research, yielding both tractability and hardness results, along with structural insights. 
Table~\ref{tab:multi-action} summarizes the current state of the art. 
Despite this progress, several fundamental questions remain open:

\begin{itemize}
\item Characterizing tractable reward functions.
Current results identify specific families of reward functions for which the optimal contract can be computed in polynomial time (e.g., additive, GS, ultra, supermodular), and others for which the problem is intractable (e.g., general submodular). A complete characterization of the tractable frontier, however, is still lacking. Pinpointing the structural properties of reward functions that drive tractability remains a central open problem. Partial progress toward this characterization has been made in~\citep{demand_working_paper}.
\item Closing approximation gaps.  
Another key challenge is to determine tight approximation guarantees for submodular, XOS, and subadditive rewards under value-oracle access (see~\Cref{tab:multi-action} for current results).
\item Beyond binary outcomes. 
A natural direction is to extend the model to settings with richer outcome spaces, going beyond the binary success/failure assumption.
\end{itemize}

\section{Setting 2: Multi-agent binary-actions}
\label{sec:model2}

A central question in contract design is how to incentivize a team of agents~\citep{Holmstrom82,holmstrom1990regulating,itoh1991incentives,legros1991efficiency}.
A combinatorial perspective on this problem was introduced by~\cite{BabaioffFN06,BabaioffFNW12},
who considered the complexity of team composition in multi-agent settings, and defined the resulting optimization problem as the combinatorial agency problem.
In this section, we focus on the multi-agent binary-action model of~\cite{DuettingEFK23}, which generalizes this framework by representing rewards as monotone set functions. Here, a principal interacts with $n$ agents, each choosing whether to exert effort, and the project's reward depends on the subset of agents who do so. State-of-the-art results for this model are summarized in~\Cref{tab:multi-agent}.

\subsection{Model}
\label{sec:model2-model}

In the multi-agent binary-action model, a principal interacts with agents $\agents=\{1,\ldots,n\}$. Each agent either exerts effort or not; exerting effort entails a cost $c_i\in\reals_{\ge 0}$ for agent $i$.
The project has two possible outcomes: success or failure. The principal's payoff is normalized to $1$ in case of success and $0$ otherwise.

A \emph{reward function} $f:2^{\agents}\to [0,1]$ maps each set $S\subseteq\agents$ of agents exerting effort to the success probability $f(S)$. As before, $f$ is assumed to be monotone and normalized.  

The principal observes only the outcome (success or failure) and incentivizes the agents via a (linear) contract
$\contract=(\alpha_1,\ldots,\alpha_n)$, where $\alpha_i$ is the fraction of the reward paid to agent $i$ upon success.
As in~\Cref{sec:model1}, restricting attention to linear contracts in this model entails no loss of generality.

\paragraph{Utility functions and equilibria.}
Given a contract $\contract$ and a set $S$ of agents exerting effort, the \emph{principal's utility} is $(1-\sum_{i\in\agents}\alpha_i)\,f(S)$, while \emph{agent $i$'s utility} is $\alpha_i f(S)-\ind{i\in S}\cdot c_i$, where $\ind{i\in S}$ is $1$ if $i\in S$ and $0$ otherwise. 
Importantly, agent $i$ is paid $\alpha_i$ upon success regardless of effort, but incurs the cost $c_i$ only when exerting effort.

The analysis focuses on (pure) \emph{Nash equilibria} (NE) of the induced game among agents. A contract $\contract$ is said to \emph{incentivize} a set $S\subseteq\agents$ if 
\begin{align*}
    &\alpha_i f(S)-c_i \;\ge\; \alpha_i f(S\setminus\{i\}) &\text{for all } i\in S,\\
    &\alpha_i f(S) \;\ge\; \alpha_i f(S\cup\{i\})-c_i &\text{for all } i\notin S.
\end{align*}

Equilibria need not be unique, so a contract is viewed as the pair $(\contract,S)$, consisting of a payment vector $\contract$ and a NE $S$ induced by $\alpha$.\footnote{Equilibrium notions beyond pure Nash equilibria are considered in~\Cref{sec:model2-beyond-ne}.} 

\paragraph{The contract design problem.}
Fix a set $S$. If $S$ is incentivizable, the minimum payment required to incentivize it is
\begin{align*}
    &\alpha_i \;=\; \frac{c_i}{f(S)-f(S\setminus\{i\})}\;=\;\frac{c_i}{f(i\mid S\setminus\{i\})} && \text{for all } i\in S,\\
    &\alpha_i \;=\; 0 && \text{for all } i\notin S,
\end{align*}
with the convention that $\frac{c_i}{f(i \mid S \setminus\{i\})}$ is zero if $c_i = 0$ and $f(i \mid S \setminus\{i\}) = 0$, and infinity when $c_i > 0$ and $f(i \mid S \setminus\{i\}) = 0$. 

The principal's objective is to maximize her expected profit, given by
\[
g(S)\;=\;\Bigl(1-\sum_i \alpha_i\Bigr) f(S)\;=\;\left(1-\sum_{i\in S}\frac{c_i}{f(i\mid S\setminus\{i\})}\right) f(S).
\]
Throughout this section, let $\sstar$ denote an optimal set of agents, i.e., a set $S$ maximizing $g(S)$.

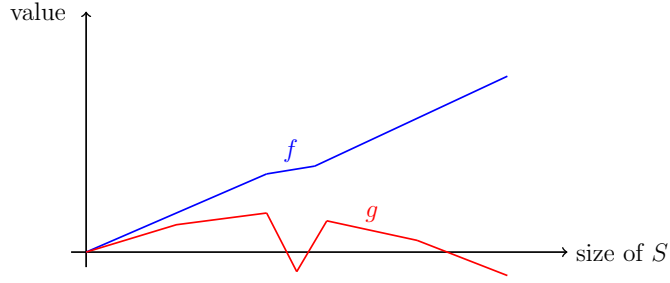
\begin{figure}
\vspace*{-10pt}
\begin{center}
\scalebox{0.8}{%
\begin{tikzpicture}
\draw[black, thick,->] (-0.25,0) -- (8,0) node[right] {\large size of $S$};
\draw[black, thick,->] (0,-0.25) -- (0,4); \node at (-0.8,4) {\large value};
\def\y{1.3}
%% f function
\draw[blue,thick,-] (0,0*\y) -- (3,1*\y);
\draw[blue,thick,-] (3,1*\y) --node[above] {\large{\textcolor{blue}{$f$}}} (3.8,1.1*\y);
\draw[blue,thick,-] (3.8,1.1*\y) -- (7,2.25*\y);
%% g function
\draw[red,thick,-] (0,0*\y) -- (1.5,0.35*\y) -- (3,0.5*\y);
\draw[red,thick,-] (3,0.5*\y) -- (3.5,-0.25*\y);
\draw[red,thick,-] (3.5,-0.25*\y) -- (4,0.4*\y);
\draw[red,thick,-] (4,0.4*\y) --node[above]{\large{\textcolor{red}{$g$}}} (5.5,0.15*\y);
\draw[red,thick,-] (5.5,0.15*\y) -- (7,-0.3*\y);
\end{tikzpicture}
}
\vspace*{-30pt}
\end{center}
\caption{
An example of an XOS reward function $f$ as a function of $|S|$, and the induced principal's utility $g$ under the best contract incentivizing $S$.}
\label{fig:f-vs-g}
\end{figure}

A core challenge is that the structural properties of $f$ need not transfer to $g$. Even when $f$ is nonnegative, monotone, and submodular, the induced $g$ can be non-monotone and may take negative values; for XOS $f$, the function $g$ can even fail to be subadditive (see~\Cref{fig:f-vs-g}).
The next example illustrates this loss of structure for two agents and a submodular reward function.
\begin{example}[Multiple agents with submodular $f$~\citep{DuettingEFK23}]
Consider $\agents=\{1,2\}$ with costs $c_1=c_2=0.25$ and a submodular reward function $f$ given by $f(\emptyset)=0$, $f(\{1\})=f(\{2\})=0.5$, and $f(\{1,2\})=0.75$.
Evaluating $g(S)$ on each $S\subseteq\agents$: (i) For $S=\emptyset$, taking $\alpha_1=\alpha_2=0$ yields $g(\emptyset)=0$. (ii) For $S=\{1\}$ (symmetrically, $\{2\}$), the optimal choice is $\alpha_1=c_1/f(\{1\})=0.5$ and $\alpha_2=0$ (respectively, $\alpha_2=0.5$ and $\alpha_1=0$), giving $g(\{1\})=g(\{2\})=(1-\alpha_1)f(\{1\})=0.25$.  (iii) For $S=\{1,2\}$, the required shares are $\alpha_1=c_1/(f(\{1,2\})-f(\{2\}))=0.25/(0.75-0.5)=1$ and similarly $\alpha_2=1$, hence $g(\{1,2\})=(1-2)f(\{1,2\})<0$.
Thus an optimal contract incentivizes exactly one agent (either $1$ or $2$). Although $f$ is monotone, nonnegative, and submodular, the induced $g$ is non-monotone and can be negative.
\end{example}

\subsection{Additive rewards}
\label{sec:model2-additive}

Moving from one agent to many agents introduces significant complexity. Even for the simplest class of additive rewards, the optimal contract problem is \textsf{NP}-hard, though it admits an FPTAS.

\begin{proposition}\citep{DuettingEFK23}
\label{prop:hardness-for-additive}
In multi-agent binary-action settings with additive rewards, computing the optimal contract is \textsf{NP}-hard; on the positive side, it admits an FPTAS.
\end{proposition}

\begin{proof}[Proof sketch]
Hardness follows via a reduction from \textsf{PARTITION}. For the approximation scheme, assume $b=\max_{i\in\sstar} f(\{i\})$ is known (one can try all $n$ possibilities). Let $\delta=\epsilon/n$ and define $\tilde f$ by rounding down each singleton value to the nearest multiple of $\delta b$: $\tilde f(\{i\})=\left\lfloor \frac{f(\{i\})}{\delta b}\right\rfloor \delta b$, and extend additively, $\tilde f(S)=\sum_{i\in S}\tilde f(\{i\})$. Hence each $\tilde f(S)$ is a multiple of $\delta b$.  
For each threshold $x=k\delta b$, let $T_x$ minimize $\sum_{i\in S} \frac{c_i}{f(\{i\})}$ among sets with $\tilde f(S)\ge x$. The algorithm outputs $T_x$ maximizing $\left(1-\sum_{i\in T_x}\frac{c_i}{f(\{i\})}\right) x$ over all $k\in\{0,1,\ldots,\lceil n/\delta\rceil\}$.

The proof shows that only polynomially many sets $T_x$ are needed, and each can be computed in polynomial time via dynamic programming. Standard arguments then show the returned set achieves a $(1-\epsilon)$-approximation to the optimal value.
\end{proof}

\subsection{Constant-approximation for submodular and XOS rewards}
\label{sec:model2-submodular}

This section turns to submodular rewards and to the broader class of XOS rewards. The main result of \cite{DuettingEFK23} is that, for both classes, polynomial-time constant-factor approximations to the optimal contract are achievable, using value oracles (for submodular) or demand oracles (for XOS).

The approach relies on an unconventional use of ``prices'' and (approximate) demand queries to identify a set of agents that the principal should seek to incentivize, together with a novel scaling property for XOS functions and their marginals that may be of independent interest.

\begin{theorem}[constant approximation for submodular and XOS rewards~\citep{DuettingEFK23}]
\label{thm:const-sm-xos}
In multi-agent binary-action settings, with \emph{submodular} reward functions $f$, there is a polynomial-time $O(1)$-approximation using a value oracle. If $f$ is \emph{XOS}, a polynomial-time $O(1)$-approximation is achievable using a demand oracle.
\end{theorem}

Recall that $\sstar$ is a maximizer of $g$. The goal is to compute a set $S$ with $g(S)\ge \text{const}\cdot g(\sstar)$.

The proof proceeds via several lemmas. The outline below assumes access to value or demand oracles for $f$; extending the result to submodular $f$ with value oracles alone requires only minor modifications and leverages known algorithms for approximate demand~\citep{SviridenkoVW17,Harshaw19a}; see~\cite{DuettingEFK23} for details.

The first lemma provides a useful decomposition of the benchmark.
Notably, this decomposition has found broader applicability, serving as a useful analytical tool in related settings (e.g., \citep{DuettingEFK24,FeldmanTPS25}).

Let $\agents' = \{ i \in \agents \mid \frac{c_i}{f(\{i\})} \leq \frac{1}{2} \}$ denote the set of ``small'' agents.

\begin{lemma}[Decomposition lemma]
\label{lem:onebigagent}
If $f$ is XOS (or more generally subadditive), then it holds that  $g(\sstar)\;\le\; f(\sstar\cap \agents') \;+\; \max\bigl\{0,\max_{i\in\agents} g(\{i\})\bigr\}$.
\end{lemma}

Since $\max_{i} g(\{i\})$ is the optimal single-agent contract and easy to compute, attention can focus on obtaining utility $\Omega\!\left(f(\sstar\cap \agents')\right)$.

\paragraph{Computing a good set via demand oracle.}
From Lemma~\ref{lem:onebigagent}, it suffices to restrict attention to the set of small agents; for exposition, suppose the optimal value $f(\sstar)$ is known (but not $S^\star$ itself); a standard guessing routine removes this assumption~\citep{DuettingEFK23}.

Two additional lemmas connect the optimal value and marginals to costs and motivate a price-based demand approach. Clearly, $\sum_{i\in\sstar} c_i\le f(\sstar)$ (otherwise, the reward cannot compensate the incurred cost); for XOS, a stronger inequality holds:

\begin{lemma}
If $f$ is XOS, then for all $S\subseteq \sstar$, \;\;$\sum_{i\in S}\sqrt{c_i}\le \sqrt{f(S)}$.
\label{lem:stronger-sumci}
\end{lemma}

Next, if all marginals are sufficiently large relative to costs, the transfers are at most half of the reward:

\begin{lemma}\label{lem:stronger-condition}
If $f$ is XOS and $S$ satisfies $f(S)>0$ and $f(i\mid S\setminus\{i\}) \;\ge\; \sqrt{2c_i\,f(S)}$ for all $i\in S$, then $g(S)\ge \tfrac12 f(S)$.
\end{lemma}

These observations suggest pricing each agent by $p_i=\tfrac12\sqrt{c_i f(\sstar)}$ and selecting a \emph{demand set} $T\in\argmax_{S}\{ f(S)-\sum_{i\in S} p_i\}$. Then, by the definition of a demand set, and~\Cref{lem:stronger-sumci}, 
$f(T)\;\ge\; f(T)-\textstyle\sum_{i\in T}p_i\;\ge\; f(\sstar)-\sum_{i\in\sstar}p_i\;\ge\;\tfrac12 f(\sstar)$. Clearly, each marginal satisfies $f(i\mid T\setminus\{i\})\ge p_i=\tfrac12\sqrt{c_i f(\sstar)}$. 

This condition almost matches the condition of~\Cref{lem:stronger-condition}, except $f(T)$ may substantially exceed $f(\sstar)$.
To address this, a novel \emph{scaling property} of XOS functions is used:

\begin{lemma}
\label{lem:xosscaling}
For any XOS $f$, set $T\subseteq\agents$, and parameters $\delta\in(0,1]$ and $0\le \downscaling < f(T)$, there is a polynomial-time algorithm, using value oracle, that computes $U\subseteq T$ with
\begin{enumerate}
    \item $(1-\delta)\downscaling \le f(U) \le \downscaling + \max_{i\in T} f(\{i\})$, and \label{eq:prop_val}
    \item $f(i\mid U\setminus\{i\}) \ge \delta\, f(i\mid T\setminus\{i\})$ for all $i\in U$. \label{eq:prop_marg}
\end{enumerate}
\end{lemma}

For submodular $f$,~\Cref{lem:xosscaling} is immediate by removing elements (in any order) until the value drops below $\Psi$.
The set obtained this way would fulfill $\downscaling \leq f(U) \leq \downscaling + \max_{i\in T} f(\{i\})$ and $f(i \mid U \setminus \{i\}) \geq f(i \mid T \setminus \{i\})$ for all $i \in U$.
For XOS, the statement requires relaxation via the parameter $\delta \in (0, 1]$.
Setting $\downscaling=\tfrac{1}{32}f(\sstar)$ and using~\Cref{lem:xosscaling} with $\delta=\tfrac12$ yields a subset $U\subseteq T$ with $f(i\mid U\setminus\{i\}) \;\ge\; \tfrac12 f(i\mid T\setminus\{i\}) \;\ge\; \tfrac12 p_i \;=\; \tfrac14 \sqrt{c_i f(\sstar)} \;\ge\; \sqrt{2c_i f(U)}$, so~\Cref{lem:stronger-condition} gives $g(U)\ge \tfrac12 f(U)\ge \tfrac{1}{128} f(\sstar)\ge \tfrac{1}{128} g(\sstar)$.

\subsection{Hardness results}
\label{sec:model2-lower-bound}

\subsubsection{Hardness results for submodular and XOS rewards}
\label{sec:model2-hardness-sm-xos}

Work on the complexity of optimal contracts in the multi-agent setting sparked a line of increasingly strong lower bounds, culminating in \citep{DuettingEFK24}. That result establishes the tightness of~\Cref{thm:const-sm-xos}, showing that, even with both value and demand oracles, any algorithm limited to a sub-exponential number of queries cannot beat a constant-factor approximation for submodular rewards.

\begin{theorem}[Submodular, hardness of approximation~\citep{DuettingEFK24}]
\label{thm:upper}
In multi-agent binary-action settings, with submodular rewards, there exists a constant $\eta>1$ such that any (possibly randomized) algorithm using a sub-exponential number of value and demand queries returns an $\eta$-approximate contract with only exponentially small probability in $n$.
\end{theorem}

This strengthens two prior lower bounds: (i) no PTAS for submodular $f$ under value-oracle access alone (unless \textsf{P}$=$\textsf{NP})~\citep{EzraFS24}; and (ii) no PTAS for XOS $f$ under demand-oracle access~\citep{DuettingEFK23}.

New ideas are required to handle submodular $f$ with both oracle types. Specifically, the construction in~\citep{DuettingEFK23} hides a single ``good'' set whose subsets appear non-attractive --- an approach that relies on the non-monotonicity of XOS marginals and does not carry over to submodular functions. Conversely,~\cite{EzraFS24} hide a small ``good'' set with high value and high marginals, whose value is significantly higher than a random set of the same size --- effective against value queries, but not against demand queries.

The proof in \citep{DuettingEFK24} instead samples a hidden set $T$ of ``good'' agents of size $k\approx n/5$ (with the remaining agents ``bad'') and reduces the value of many subsets so that $T$ becomes cheaper to incentivize and thus more attractive for the principal. Any $\eta$-approximation must incentivize a set of size at most $k$ containing at least $k/2$ good agents; however, only exponentially few such sets exist relative to $\binom{n}{k}$. This yields an exponential lower bound under value queries. For the constructed instances, demand queries can be simulated with polynomially many value queries, extending the hardness to algorithms with both oracle types.


\begin{table}[t]
\fontsize{10pt}{12pt}
\selectfont
\scalebox{0.9}{
\begin{tabular}{|
>{\columncolor[HTML]{C0C0C0}}c |c
>{\columncolor[HTML]{EFEFEF}}c |
>{\columncolor[HTML]{EFEFEF}}c 
>{\columncolor[HTML]{FFFFC7}}c |}
\hline
\textbf{\begin{tabular}[c]{@{}c@{}}Multiple\\ agents\end{tabular}}    & \multicolumn{2}{c|}{\cellcolor[HTML]{C0C0C0}\textbf{Value Oracle}}                                                                                                                                                          & \multicolumn{2}{c|}{\cellcolor[HTML]{C0C0C0}\textbf{Value and Demand Oracle}}                                                                                                                                     \\ \hline
\multicolumn{1}{|l|}{\cellcolor[HTML]{C0C0C0}}                        & \multicolumn{1}{c|}{\cellcolor[HTML]{C0C0C0}\textbf{\begin{tabular}[c]{@{}c@{}}Upper bound \\ (pos)\end{tabular}}} & \cellcolor[HTML]{C0C0C0}\textbf{\begin{tabular}[c]{@{}c@{}}Lower bound\\ (neg)\end{tabular}}           & \multicolumn{1}{c|}{\cellcolor[HTML]{C0C0C0}\textbf{\begin{tabular}[c]{@{}c@{}}Upper bound \\ (pos)\end{tabular}}} & \cellcolor[HTML]{C0C0C0}\textbf{\begin{tabular}[c]{@{}c@{}}Lower bound\\ (neg)\end{tabular}} \\ \hline
\textbf{Additive}                                                     & \multicolumn{1}{c|}{\cellcolor[HTML]{FFFFC7}\begin{tabular}[c]{@{}c@{}}FPTAS\\ \cite{DuettingEFK23}\end{tabular}}                & \begin{tabular}[c]{@{}c@{}}OPT is\\ NP-hard\end{tabular}                                               & \multicolumn{1}{c|}{\cellcolor[HTML]{EFEFEF}FPTAS}                                                                 & \begin{tabular}[c]{@{}c@{}}OPT is\\ NP-hard\\ \cite{DuettingEFK23}\end{tabular}                            \\ \hline
\textbf{GS}
& \multicolumn{1}{c|}{\cellcolor[HTML]{EFEFEF}\begin{tabular}[c]{@{}c@{}}Constant \\ approx\end{tabular}}            & \begin{tabular}[c]{@{}c@{}}OPT is\\ NP-hard\end{tabular}                                               & \multicolumn{1}{c|}{\cellcolor[HTML]{EFEFEF}\begin{tabular}[c]{@{}c@{}}Constant\\ approx\end{tabular}}             & \cellcolor[HTML]{EFEFEF}\begin{tabular}[c]{@{}c@{}}OPT is\\ NP-hard\end{tabular}             \\ \hline
%\textbf{Submodular}    
\textbf{\begin{tabular}[c]{@{}c@{}}Sub-\\ modular\end{tabular}}& \multicolumn{1}{c|}{\cellcolor[HTML]{FFFFC7}\begin{tabular}[c]{@{}c@{}}Constant\\ approx\\ \cite{DuettingEFK23}\end{tabular}}    & 
\multicolumn{1}{c|}{\cellcolor[HTML]{EFEFEF}\begin{tabular}[c]{@{}c@{}}No PTAS\\ 
\cite{EzraFS24}\\
\cite{DuettingEFK24}
\end{tabular}}  
& \multicolumn{1}{c|}{\cellcolor[HTML]{EFEFEF}\begin{tabular}[c]{@{}c@{}}Constant\\ approx\end{tabular}}             & 
\cellcolor[HTML]{FFFFC7}\begin{tabular}[c]{@{}c@{}}No PTAS\\ 
\cite{DuettingEFK24}\end{tabular}  
\\ \hline
\textbf{XOS}                                                          & \multicolumn{1}{c|}{\cellcolor[HTML]{EFEFEF} $O(n)$-approx}                                                                      & \cellcolor[HTML]{FFFFC7}\begin{tabular}[c]{@{}c@{}}No better\\ than $\Omega(n^{1/6})$\\ \cite{EzraFS24}\end{tabular} & \multicolumn{1}{c|}{\cellcolor[HTML]{FFFFC7}\begin{tabular}[c]{@{}c@{}}Constant\\ approx\\ \cite{DuettingEFK23}\end{tabular}}    & \begin{tabular}[c]{@{}c@{}}No PTAS\\ \cite{DuettingEFK23}\end{tabular}                                     \\ \hline
%\textbf{Subadditive}       
\textbf{\begin{tabular}[c]{@{}c@{}}Sub-\\ additive\end{tabular}}& 
\multicolumn{1}{c|}{\cellcolor[HTML]{FFFFC7}\begin{tabular}[c]{@{}c@{}}$O(n)$-approx\\
\cite{DuttingFT24-survey}
\end{tabular}}
& 
\begin{tabular}[c]{@{}c@{}}No better\\ than $\Omega(n^{1/6})$\end{tabular}                                  & \multicolumn{1}{c|}{\cellcolor[HTML]{EFEFEF} $O(n)$-approx}                                                                      & \begin{tabular}[c]{@{}c@{}}No better\\ than $\Omega(n^{1/2})$\\ \cite{DuettingEFK23}\end{tabular}               \\ \hline \hline
\textbf{\begin{tabular}[c]{@{}c@{}}Super-\\ modular\end{tabular}}     & 
\multicolumn{1}{c|}{\cellcolor[HTML]{FFFFFF}}                                                                     & \begin{tabular}[c]{@{}c@{}}No constant\\ approx\end{tabular}                                           & \multicolumn{1}{c|}{\cellcolor[HTML]{FFFFFF}}                                                                              & \begin{tabular}[c]{@{}c@{}}No constant\\ approx\\(if \textsf{P}$\neq$\textsf{NP})\\ {\footnotesize\citet{VuongDPP23}}\end{tabular}                      \\ \hline
\end{tabular}
}
\caption{
Summary of approximation guarantees in the multi-agent binary-action setting.
Left: value-oracle access. Right: value- and demand-oracle access.
For each oracle model, we list known upper bounds (positive results) and lower bounds (hardness results) on achievable approximation ratios for different classes of reward functions.
Yellow entries indicate the main results; gray entries follow by closure (positive results extend upward and to the right, while negative results extend downward and to the left).
}
\label{tab:multi-agent}
\end{table}

For XOS rewards, \cite{EzraFS24} prove that no algorithm can do better than $O(n^{1/6})$-approximation with poly-many value queries, thus establishing a separation between the power of value and demand queries for XOS rewards. 

\subsubsection{Hardness results for subadditive and supermodular rewards}
\label{sec:model2-hardnesssubadd-supermod}

For the broader class of subadditive functions, \cite{DuettingEFK23} show that no algorithm can guarantee better than $\Omega(\sqrt{n})$ approximation. On the positive side, a simple $O(n)$-approximation is obtained with value oracles by selecting the best single-agent contract $\max_{i\in\agents} g(\{i\})$: it is shown that, for any $S$ with $g(S)\ge 0$, one has $g(S)\le \sum_{i\in S} g(\{i\}) \le n\cdot \max_{i\in S} g(\{i\})$. To see this, note that for any $i \in S$, $f(S) \leq f(S\setminus \{i\}) + f(\{i\})$, so $\frac{c_i}{f(S)-f(S\setminus{i})} \geq \frac{c_i}{f(\{i\})}$. Moreover, $f(S) \leq \sum_{i\in S}f(\{i\})$. So, $g(S) = (1-\sum_{i \in S} \frac{c_i}{f(S)-f(S\setminus\{i\})}) f(S) \leq \sum_{i\in S} (1-\sum_{i \in S} \frac{c_i}{f(\{i\})})f(\{i\}) \leq \sum_{i\in S} (1- \frac{c_i}{f(\{i\})})f(\{i\}) = \sum_{i \in S}g(\{i\})$.

For supermodular rewards, no poly-time $O(1)$-approximation, nor additive FPTAS, exists, unless \textsf{P}$=$\textsf{NP}~\citep{VuongDPP23}.

\subsection{Extensions and related models}
\label{sec:model2-open}

\subsubsection{Budget constraints and additional objective functions}
\label{sec:model2-budgets}
\cite{FeldmanTPS25,Aharoni0T25} study the multi-agent binary-action model with XOS reward functions under a budget constraint on the principal, where total payments to agents cannot exceed budget $B \in [0,1]$. Their work also broadens the scope of objectives beyond the principal's utility: 
They introduce a wide class of objectives, termed BEST (BEyong STandard) objectives, which include the principal's utility, expected reward, and social welfare, and show that an algorithm for any BEST objective under any budget can be converted into one for any other BEST objective under any other budget with only a constant-factor loss.
Together with the constant-factor approximation for unbudgeted principal's utility from \citep{DuettingEFK23}, this yields $O(1)$-approximations for many natural objectives under arbitrary budget constraints. 
Finally, the notion of the \emph{price of frugality} is introduced, which quantifies the efficiency loss due to budget restrictions, and is shown to scale asymptotically as $\Theta(1/B)$, potentially up to the cost of incentivizing a single agent.

In a related variant, \cite{gong4581020principal} study a principal who is constrained by contracting with at most $k$ agents and give a poly-time constant-factor approximation for the principal's utility, under XOS $f$, with value and demand oracles.

\subsubsection{Beyond pure Nash equilibrium}
\label{sec:model2-beyond-ne}

Studies on multi-agent contracts has focused primarily on pure Nash equilibria, but this focus is not without loss of generality: the principal may strictly benefit from recommending more complex (possibly correlated) distributions over actions while preserving incentives. This fact was already observed by \cite{BabaioffFN10}, who presented an instance with (essentially) a submodular reward function, where the principal can strictly benefit by inducing a mixed Nash equilibrium.
\cite{mixed_eq_working_paper} initiate a more elaborate study, comparing the principal's utility across pure, mixed, correlated (CE), and coarse-correlated (CCE) equilibria. 
We defer the presentation of these results to~\Cref{sec:model3-beyond-ne}, since most of them extend to the more general setting that goes beyond binary actions.

\subsubsection{Fairness considerations}
\label{sec:model2-fairness}

Fairness considerations in contract design often manifest through de-facto constraints --- pay-ratio disclosures and surtaxes, limits on allowable or deductible compensation (e.g., in procurement and tax law), and firm-level pay-equity policies.
\cite{fair-contracts-working} study fairness in the multi-agent, binary-action model with submodular rewards, imposing either a $\gamma$-equitable-pay constraint (all agents receive the same success payment, up to a factor of $\gamma$) or a $\gamma$-equitable-utility constraint (all agents receive the same utility, up to a factor of $\gamma$), where $\gamma$ is some constant. They formalize the \emph{price of fairness} --- the ratio between the principal's optimal utility with and without the constraint --- and establishes a tight bound of $\Theta(\ln n/\ln\ln n)$ on this ratio: the lower bound already holds for additive rewards, while the matching upper bound extends to submodular rewards.

Their main algorithmic result is a constant factor approximation to the optimal $\gamma$-equitable-pay and to the optimal $\gamma$-equitable-utility contracts, for any $\gamma$. 
They first present an 
$O(1)$-approximation algorithm for the optimal 
1-equitable-pay contract, and then prove that the utilities from the optimal 1-equitable-pay contract and the optimal $\gamma$-equitable-pay contract differ by only a constant factor.
To approximate the optimal $\gamma$-equitable-utility contract, they show that any $\gamma$-equitable-pay contract can be transformed into a $\gamma$-equitable-utility contract, and vice versa, in polynomial-time with value queries, while only losing a constant factor in the principal's utility.

\subsubsection{Observable individual outcomes}
\label{sec:model2-observable-outcomes}

\cite{CastiglioniM023} and \cite{GoelW2025} analyze a distinct multi-agent framework where each agent's effort produces its own observable outcome, and contracts can condition payments on individual outcomes rather than on a single aggregated project outcome. \cite{CastiglioniM023} explore the algorithmic tractability of designing contracts when the principal's reward exhibits structural properties capturing complementarities across agents, or diminishing returns. They establish hardness in the general case but identify families of reward functions that admit efficient approximation schemes. For observable individual outcomes, \cite{GoelW2025} examine a budget-constrained environment, in which the principal must distribute a fixed reward budget among agents. They prove that optimal contracts in this model belong to the class of \emph{Luce} contracts, where each agent is assigned a weight and the available budget is divided proportionally among the successful agents with the highest weights. This characterization narrows the search space for optimal solutions and ensures desirable properties such as minimal variance in payments among feasible contracts.

\subsubsection{Multiple projects}
\label{sec:multi-project}

\cite{AlonCETLT25} introduce a \emph{multi-project} model, where a principal allocates $n$ agents across multiple projects. Each agent may join at most one project and incurs project-specific costs, while each project has a binary outcome (success or failure) determined by a combinatorial success function mapping the set of participating agents to a success probability. Their main result is a polynomial-time constant-factor approximation: for submodular rewards using value oracles, and for XOS rewards using both value and demand oracles. The algorithm combines bipartite matching for projects with ``dominant'' agents and an LP-based approach with a novel approximate demand oracle for projects without dominant agents.

\subsection{Open problems}
\label{sec:model2-open-problems}
Several open questions remain in the multi-agent binary-action setting, including:
\begin{itemize}
    \item PTAS/FPTAS for GS rewards: Does the optimal contract problem admit a PTAS (or FPTAS) under GS rewards? In particular, can the FPTAS for additive rewards be adopted to yield a $(1+\epsilon)$-approximation for GS, or even for meaningful subclasses such as OXS? At present, the hardness result that rules out a PTAS is for submodular rewards alone.
    \item Tight bounds for subadditive rewards: What is the best achievable approximation under subadditive rewards? The current bounds leave a gap between a lower bound of $\Omega(\sqrt{n})$ and an upper bound of $O(n)$.
    \item Beyond binary outcomes: Extending the analysis beyond binary (success/failure) outcomes remains open, paralleling the single-agent multi-action discussion in~\Cref{sec:model1-open}.
\end{itemize}

\section{Combined setting: multi-agent multi-actions}
\label{sec:model3}

A natural generalization of the previous models is obtained by allowing multiple agents, each capable of taking multiple actions. This setting, introduced in~\citep{DuettingEFK24}, combines the challenges of both the multi-agent binary-action model and the single-agent multi-action model, and is referred to as the multi-agent multi-action model. The resulting framework presents qualitatively new challenges, since it lacks structural properties that were central to the tractability of the special cases.

\subsection{Model}
\label{sec:model3-model}   

In the multi-agent multi-action model, a principal interacts with $n$ agents. Each agent $i$ chooses a subset $S_i \subseteq \actions_i$ of available actions, and the overall profile of actions is $S = \bigcupdot_i S_i$. Agents may also choose the empty set. Each action $j$ has a cost $c_j$ to the corresponding agent, and costs are additive across actions.  

The outcome of the project is either success or failure, with the principal's reward normalized to $1$ in case of success and $0$ otherwise. The success probability is determined by a monotone function $f(S)$ over the chosen actions of all agents. As before, we refer to $f$ as the reward function.

The principal cannot observe the specific actions, only the final outcome, and therefore uses a linear contract $\contract = (\alpha_1,\ldots,\alpha_n)$ specifying reward shares to each agent $i$ conditional on success. Given a contract $\contract$ and an action profile $S$, the expected utility of agent $i$ is $\alpha_i f(S) - c(S_i)$, where $S_i = S \cap \actions_i$.  

A profile $S$ constitutes a (pure) Nash equilibrium under contract $\contract$ (we also say that $S$ is induced or incentivized by $\contract$) if, for every agent $i$, the chosen set $S_i \subseteq \actions_i$ maximizes $i$’s utility given the actions of the others, $S_{-i} = \bigcupdot_{j \in \agents \setminus {i}} S_j$; that is, $S_i \in \argmax_{T \subseteq \actions_i} \{\alpha_i f(T,S_{-i}) - c(T)\}$.

The principal's objective is to design a contract that maximizes her expected utility, $(1-\sum_i \alpha_i) f(S)$, where $S$ is a Nash equilibrium of $\contract$.

\paragraph{Challenges.}
This model unifies the single-agent multi-action ($n=1$) and multi-agent binary-action ($\actions_i = \{i\}$) settings. Each special case involves searching over exponentially many candidate sets, but the character of the resulting challenges varies across the models. 
In the former, one of the main challenges lies in determining which sets of actions can be incentivized, while in the latter, all subsets of agents can be incentivized but the problem lies in finding the optimal one. In both cases, monotonicity properties play a crucial role: in the single-agent setting, success probability increases with $\alpha$, and in the multi-agent binary-action setting with submodular rewards, dropping an agent does not discourage others from exerting effort.  

The combined model does not preserve these monotonicity properties. Deciding whether a given action profile is implementable becomes substantially harder, and even if it is, it is not clear what is the best way to incentivize it.
In particular, the combined model loses the monotonicity properties that were instrumental in solving the special cases. For example, lowering the payment to one agent can lead another agent to reduce effort, even under submodular rewards~\citep{DuettingEFK24}.  

Another challenge is the possibility of multiple equilibria with potentially large gaps between them. Unlike the special cases, where equilibria are essentially unique (for submodular $f$, under mild assumptions), a single contract here may admit multiple equilibria. In certain submodular examples, the best equilibrium of the best contract provides significantly higher utility than the worst equilibrium, yet no contract guarantees more than the lower utility across all equilibria~\citep{DuettingEFK24}. This establishes intrinsic limits on approximation guarantees if robustness against all equilibria is required.

\subsection{Equilibrium existence}
\label{sec:model3-eq-existence}

Despite these challenges, every contract in this model admits at least one equilibrium, since the induced game among the agents is a potential game~\citep{VuongDPP23,DuettingEFK24}. A potential function can be defined as $\phi(S) = f(S) - \sum_{i \in \agents} c(S_i)/\alpha_i$ (with the convention that if $\alpha_i = 0$, $c(S_i)/\alpha_i = \infty$ when $c(S_i) > 0$, and $c(S_i)/\alpha_i = 0$ when $c(S_i) = 0$). Hence, existence of pure Nash equilibria is guaranteed. 
Moreover, the specific form of the potential function implies that an equilibrium $S$ can be identified using a single demand query to $f$.

\begin{proposition}[Potential game~\citep{VuongDPP23,DuettingEFK24}]
\label{prop:potential-function}
    In any multi-agent multi-action setting, every contract $\contract$ admits at least one pure Nash equilibrium.
\end{proposition}

\subsection{Constant-approximation for submodular rewards, with value and demand oracles}
\label{sec:model3-const}

A central result of~\cite{DuettingEFK24} is that, for submodular reward functions, there exists a polynomial-time algorithm, with value and demand oracles, that computes a contract yielding a constant-factor approximation to the optimal principal's utility. Remarkably, this guarantee holds for \emph{all} equilibria of the contract, ensuring robustness to equilibrium selection.

\begin{theorem}[Constant approximation for submodular rewards, with value and demand oracles~\citep{DuettingEFK24}]
\label{thm:const-approx-multi-multi}
In multi-agent multi-action settings, with submodular reward functions and access to value and demand oracles, there exists a polynomial-time algorithm that computes a contract $\contract$ such that every equilibrium of $\contract$ yields a constant-factor approximation to the optimal principal's utility.
\end{theorem}

\begin{remark}
Since demand queries for GS functions can be simulated using poly-many value queries (see~\Cref{thm:gs-greedy}), \Cref{thm:const-approx-multi-multi} also implies a constant-factor approximation for GS rewards with value oracles alone. Notably, this problem is known to be $\mathsf{NP}$-hard even for additive $f$ with binary actions (see~\Cref{sec:model2-additive}).
\end{remark}

The proof combines several new ideas, two of which are \emph{subset stability} and the \emph{doubling lemma}. Subset stability relaxes Nash equilibrium by requiring only that agents gain no benefit from deviating to subsets of their chosen actions (even if other types of deviations may still be profitable).

\begin{definition}[Subset stability~\citep{DuettingEFK24}]
\label{def:subset-deviations}
A set of actions $S$ is \emph{subset stable} with respect to a contract $\contract$ if, for all agents $i$ and all $S'_i \subseteq S_i$,  it holds that 
$\alpha_i f(S_i,S_{-i}) - c(S_i) \;\geq\; \alpha_i f(S'_i,S_{-i}) - c(S'_i)$.
\end{definition}

% The doubling lemma establishes that if $S$ is subset stable under $\contract$, then the reward of  any equilibrium of the scaled contract $2\contract+\vec{\epsilon}$ is at  least half of $f(S)$.

The doubling lemma states that if $S$ is subset stable under $\contract$, then any equilibrium of the scaled contract $2\contract+\vec{\epsilon}$ attains reward at least $f(S)/2$.

\begin{lemma}[Doubling lemma~\citep{DuettingEFK24}]
\label{lemma:anyequilibrium}
If $f$ is submodular, and $S$ is subset stable with respect to a contract $\contract$, then for every $\epsilon>0$, any equilibrium $\Seq{}$ of $2\contract+\vec{\epsilon}$ satisfies $f(\Seq{}) \geq \tfrac{1}{2}f(S)$, where $\vec{\epsilon}=(\epsilon,\ldots,\epsilon) \in \reals_+^n$. % Suppose $f$ is submodular. Let $\epsilon > 0$ and let $\vec{\epsilon} = (\epsilon,\ldots,\epsilon) \in \reals_+^n$. Let $S$ be a subset-stable action set with respect to a contract $\contract$.
\end{lemma}

\begin{proof}
[Proof of~\Cref{thm:const-approx-multi-multi}] 
The proof reduces the analysis to two cases: either no agent is large in the optimal contract, or a single agent is incentivized. While a related decomposition was employed in~\citep{DuettingEFK23}, its extension beyond binary actions is far from immediate, due to the lack of monotonicity mentioned above.

{\bf Case 1:} No agent is ``large'' in the optimal contract; namely, for the optimal contract $\contract^\star$ and the optimal equilibrium $\sstar$ of $\contract^\star$, $f(\sstar_i)$ is bounded away from $f(S^\star)$ for every $i \in \agents$.  

In this case, an algorithm is provided that, with value oracle access, computes a contract $\contract$ together with a subset-stable set $S$ satisfying $\sum_{i \in \agents} \alpha_i \leq \tfrac{1}{4}$ and achieving $f(S) = \Omega(1)\cdot f(\sstar)$.
Then, the doubling lemma shows that $2 \contract + \vec{\epsilon}$ is a good approximation of the optimal contract. 

Notably, although subset stability is easier to guarantee than full stability, verifying it directly still involves checking exponentially many conditions. To address this, \cite{DuettingEFK24} identify a sufficient condition for subset stability, which uncovers an
interesting connection to a particular form of bundle prices. That is, it suffices to select a set that maximizes $f(S) - \gamma \sum_{i \in \agents} \sqrt{\sum_{j \in S_i} c_i}$ for some $\gamma >0$. This objective can be interpreted as computing a demand set with respect to specific non-additive bundle prices. For this latter problem, an approximation algorithm is provided.

{\bf Case 2:} Only one agent is incentivized in the optimal contract $\contract^\star$; that is, $\alpha^\star_i = 0$ for all but one agent $i$. 
Recall that in the single-agent case, there exists an FPTAS even for general $f$, under demand oracle access (see~\Cref{thm:FPTAS-monotone}). 
In the multi-agent setting, other equilibria may emerge, involving additional agents (beyond the intended one) taking actions.
By appealing once more to the concept of subset stability and the doubling lemma, it can be shown that the loss to the principal's utility due to these agents can be bounded.

The final step shows that the principal's optimal utility can be approximated within a constant factor by focusing on the two identified cases. In particular, it demonstrates that when the optimal solution involves a large agent, it is sufficient to incentivize a single agent. Establishing this is nontrivial due to the absence of monotonicity noted earlier. The argument relies on two structural properties: every equilibrium is subset stable, and subset stability is preserved under the removal of actions, together with an additional application of the doubling lemma.
\end{proof}

\subsection{Extensions and related models}
\label{sec:model3-extensions}

\subsubsection{Budget constraints and additional objectives}
\label{sec:model3-budgets}

\cite{budget-multi-working} extend the budgeted variant of the multi-agent binary-action model of~\cite{FeldmanTPS25} to the multi-agent, multi-action setting. This framework combines two sources of complexity: combinatorial actions and budget constraints. They show that for submodular rewards, no randomized polynomial-time algorithm can approximate the optimal budget-feasible principal's utility within any finite factor, even under demand-oracle access.

Interestingly, either source of complexity alone admits a constant-factor approximation: (i) the multi-agent binary-action setting with submodular rewards and budgets (see~\Cref{sec:model2-budgets}), and (ii) the unbudgeted multi-agent multi-action setting with submodular rewards (\Cref{thm:const-approx-multi-multi}). Thus, it is the combination of the two that leads to intractability.

On the positive side, when rewards are GS, a deterministic polynomial-time constant-approximation is obtained using only value queries. These results highlight a separation between budgeted and unbudgeted settings in combinatorial contracts and delineate GS as a tractable frontier for budgeted contract design in multi-agent multi-action settings.

Finally, for additive rewards they present an FPTAS, showing that arbitrary approximation is feasible under any budget. This is the first FPTAS for multi-agent multi-action settings, even without budget constraints.

\subsubsection{Beyond pure Nash equilibrium}
\label{sec:model3-beyond-ne}
As discussed in~\Cref{sec:model2-beyond-ne}, \cite{mixed_eq_working_paper} offer a comprehensive comparison of the principal's utility under different equilibrium notions, including pure, mixed, correlated (CE), and coarse-correlated (CCE) equilibria.
For submodular and XOS rewards, they show that moving to richer equilibrium concepts (up to CCE) can increase utility by at most a constant factor, whereas for subadditive rewards the gap can be polynomial in the number of agents. Supermodular rewards exhibit a further distinction between binary and combinatorial actions: with binary actions there is no gap between pure NE and CCE, but with combinatorial actions a separation emerges: correlated equilibria yield the same utility as pure equilibria, while coarse-correlated equilibria can achieve arbitrarily higher utility. Finally, for general monotone rewards, the gap can be arbitrarily large, even between mixed and pure equilibria.

\subsubsection{Multiple outcomes and randomized contracts}
\label{sec:multi-agent-randomized}

\cite{CacciamaniEtAl24} study a general, explicitly represented, multi-agent multi-action model with multiple outcomes, where each action profile induces a probability distribution over outcomes. They introduce a class of \emph{randomized contracts}, consisting of a distribution over recommended action profiles, together with a classic contract for each agent for each action profile. Their equilibrium concept resembles correlated equilibrium. 
They show that the gap between
randomized and deterministic contracts in this setting can be unbounded. 
On the algorithmic side, they cast the optimal randomized contract problem as a quadratic program and design a scheme that computes $(1+\epsilon)$-approximate randomized contracts in polynomial time. 

\subsection{Open problems}
\label{sec:model3-open}

The multi-agent multi-action model raises several natural questions for future work:

\begin{itemize}

\item Extension to XOS rewards.
A central open question is whether the positive results for submodular rewards extend to XOS. The current analysis does not apply even with binary actions, though constant-factor approximation guarantees might still be possible under weaker benchmarks; for instance, by coupling contracts with recommended equilibria \citep{DuettingEFK23}, or by evaluating contracts against their worst-case equilibrium. Since the proofs rely heavily on submodularity (e.g., via the doubling lemma), extending them to XOS will require new techniques.

\item Approximation schemes for GS.
Another important question is whether the problem admits a PTAS or FPTAS when rewards are GS. This remains open even in the binary-action case, though a hardness result (if exists) may be easier to obtain for multiple actions.

\item Restricted domains. Restricting the domain to few actions or few agents raises new possibilities. With a constant number of actions, can submodular rewards admit constant-factor approximations with value queries alone, as in the binary-action case \citep{DuettingEFK23}, unlike the hardness under an arbitrary number of actions \citep{EzraFS24}? With a constant number of agents, can we achieve exact optimization for GS rewards, PTAS/FPTAS for submodular rewards, or constant-factor guarantees for subadditive rewards (with demand oracles)?
\end{itemize}

\section{Concluding remarks and future directions}
\label{sec:conclusions}

This article explores the emerging field of combinatorial contract design, highlighting models in which agents face exponentially many actions or interact strategically in complex teams. We reviewed structural insights, algorithmic techniques, and hardness barriers across three central settings. These studies reveal both surprising tractable cases and strong impossibility results, showing how structural properties of reward functions, together with considerations such as budget and fairness, shape computational guarantees.

A key limitation of the current literature is its near-exclusive focus on binary outcomes. The classic model of \cite{Holmstrom79} and its algorithmic counterpart by \cite{DuttingRT19} consider multiple outcomes, a feature that is both more realistic and conceptually richer. Some results surveyed here extend to multiple outcomes, but only under linear contracts, where this restriction is no longer without loss of generality. Extending the combinatorial settings to richer outcome spaces remains an important open direction. It is still unclear how multiple outcomes should be modeled in these settings, and this challenge is largely unexplored. The introduction of multiple outcomes also opens several research avenues, such as understanding the power of ambiguous contracts \citep{DuettingFP23,DuettingFR25}, and quantifying the effectiveness of simple contracts \citep{DuttingRT19}. In a complementary direction, \cite{DuttingRT21} propose a model in which the outcome space itself is combinatorial (rather than the agents or actions).

Much remains to be explored. Immediate directions include characterizing the frontier of tractable reward functions and tightening existing approximation bounds. Broader directions include extending these models to Bayesian environments with private information~\citep{GuruganeshSW21,AlonDT21}, or to learning-based approaches under partial information~\citep{HoSV16,ZhuEtAl22}. Another promising line of work is incorporating inspection and monitoring of actions \citep{EzraLR24} into these combinatorial settings. Importantly, each of these directions is already highly non-trivial in the simpler non-combinatorial setting.

Looking ahead, emerging AI agents raise a fresh set of challenges and opportunities for contract design. Their growing computational capabilities, their ability to operate over vast and evolving action spaces, and their departure from classical assumptions about humans call for new models and analysis.
Taken together, the directions outlined in this article chart a new frontier at the intersection of computation and economics.

\bibliographystyle{plainnat}
\bibliography{icm_references}
\end{document}